\documentclass[11pt]{amsart}
\usepackage{amssymb,mathrsfs,graphicx,enumerate}
\usepackage{amsmath,amsfonts,amssymb,amscd,amsthm,bbm}
\usepackage[retainorgcmds]{IEEEtrantools}
\usepackage{colortbl}

\title[Emergent behaviors of the generalized Lohe matrix model]{Emergent behaviors of the generalized Lohe matrix model}

\author[Ha]{Seung-Yeal Ha}
\address[Seung-Yeal Ha]{\newline Department of Mathematical Sciences\newline Seoul National University, Seoul 08826, Korea, and \newline
Korea Institue for Advanced Study, Hoegiro 85, 02455, Seoul, Korea}
\email{syha@snu.ac.kr}

\author[Park]{Hansol Park}
\address[Hansol Park]{\newline Department of Mathematical Sciences\newline Seoul National University, Seoul 08826, Korea}
\email{hansol960612@snu.ac.kr}

\newtheorem{theorem}{Theorem}[section]
\newtheorem{lemma}{Lemma}[section]
\newtheorem{corollary}{Corollary}[section]
\newtheorem{proposition}{Proposition}[section]
\newtheorem{remark}{Remark}[section]

\newtheorem{definition}{Definition}[section]

\newcommand{\bbr}{\mathbb R}

\newcommand{\bbu}{\mathbb U}

\newcommand{\bbc}{\mathbb C}

\begin{document}

\date{\today}


\thanks{\textbf{Acknowledgment.} The work of S.-Y. Ha is supported by NRF-2017R1A2B2001864.}

\begin{abstract}
We present a first-order aggregation model on the space of complex matrices which can be derived from the Lohe tensor model on the space of tensors with the same rank and size. We call such matrix-valued aggregation model as ‘‘{\it the generalized Lohe matrix model}". For the proposed matrix model with two cubic coupling terms, we study several structural properties such as the conservation laws, solution splitting property. In particular, for the case of only one coupling, we reformulate the reduced Lohe matrix model into the Lohe matrix model with a diagonal frustration, and provide several sufficient frameworks leading to the complete and practical aggregations. For the estimates of collective dynamics, we use a nonlinear functional approach using an ensemble diameter which measures the degree of aggregation.
\end{abstract}

\maketitle \centerline{\date}


\section{Introduction} \label{sec:1}
\setcounter{equation}{0}
Collective behaviors of many-body systems often appear in biological, chemical and physical complex systems, e.g., flocking of birds, swarming of fish, an array of Josephson junctions, etc, (see related literature \cite{A-B, A-B-F, B-H, B-T1, B-T2, B-B,H-K-P-Z, Pe, P-R, St, T-B-L, T-B, VZ, Wi1}). Despite of its ubiquity in our nature, modeling-based studies were first begun by Winfree and Kuramoto \cite{Ku1, Ku2, Wi2, Wi1} in a half century ago. After their pioneering works, several agent-based models for collective dynamics were proposed and extensively studied in applied mathematics, control theory and nonlinear dynamics communities \cite{B-C-S, D-F-M-T, D-F-M, De, H-K, H-R}. 

In this paper, we propose a generalized first-order aggregation model on the space of complex matrices with the same size and study its emergent dynamics.  For the square matrices, our proposed model can be reduced to the Lohe matrix model  \cite{H-K-R, H-R, Lo-0, Lo-1, Lo-2} on the unitary matrix group. Since the space of matrices can be regarded as the space of rank-2 tensors, we begin with a general framework  introduced in \cite{H-P1, H-P2} on the space of tensors, and derive our model as a special case of the general model (see Section \ref{sec:2}).   

Let ${\mathbb C}^d$ and ${\mathbb C}^{d_1 \times d_2}$ be a $d$-dimensional complex vector space which is isomorphic to $\bbr^{2d}$ and a complex vector space consisting of $d_1 \times d_2$ matrices with complex entries, respectively. We set ${\mathcal T}_2(\bbc; d_1 \times d_2)$ to be the vector space of rank-2 complex tensors with size $d_1 \times d_2$. Let $T$ be a rank-2 tensor with size $d_1 \times d_2$, and its $(\alpha, \beta)$-component of $T$ is denoted by $[T]_{\alpha \beta}$. Then, its hermitian conjugation of $T$ is denoted by $T^*$ and 
\begin{equation} \label{A-0}
[T^*]_{\alpha \beta} := [{\bar T}]_{\beta \alpha}, \quad \alpha = 1, \cdots, d_2,~~\beta = 1, \cdots, d_1,
\end{equation}
and let $A \in {\mathcal T}_4(\bbc; d_1 \times d_2 \times d_1 \times d_2)$ be the skew-Hermitian rank-4 tensor:
\begin{equation} \label{A-0-1}
[{\bar A}]_{\alpha \beta \gamma \delta} = -[A]_{\gamma \delta \alpha \beta}, \quad 1 \leq \alpha, \gamma \leq d_1,~~1 \leq \beta, \delta \leq d_2.
\end{equation}
With the above preparations, we are now ready to present our first-order aggregation model. \newline

Consider an ensemble $\{T_i \}_{i=1}^{N}$ of rank-2 tensors, and we set $T_c := 1/N \sum_{i=1}^{N} T_i$. Then, our proposed first-order aggregation on ${\mathcal T}_2(\bbc: d_1 \times d_2)$ is given as a mean-field form.
\begin{equation} \label{A-1}
{\dot T}_i =A_i T_i +\kappa_{1}(T_cT_i^* T_i -T_i T_c^* T_i)+\kappa_{2}(T_iT_i^* T_c-T_i T_c^* T_i), \quad 1 \leq i \leq N, \\
\end{equation}
where  $\kappa_{1}$ and $\kappa_{2}$ are nonnegative coupling strengths, and the products in cubic interaction term are standard matrix multiplications, and $A_j \in {\mathcal T}_4(\bbc;d_1 \times d_2 \times d_1 \times d_2)$ is rank-4 tensor defined by 
\begin{equation} \label{A-1-1}
[A_iT_i]_{\alpha\beta} :=[A_i]_{\alpha\beta\gamma\delta}[T_i]_{\gamma\delta}, \quad 1 \leq \alpha, \gamma \leq d_1,~~1 \leq \beta, \delta \leq d_2,
\end{equation}
where we used Einstein summation convention. Due to the relations \eqref{A-0} and \eqref{A-0-1}, it is easy to see that linear term and the coupling terms clearly in the R.H.S. of \eqref{A-1} belong to ${\mathcal T}_2(\bbc; d_1 \times d_2)$. Thus, model \eqref{A-1} generates a dynamical system on ${\mathcal T}_2(\bbc; d_1 \times d_2)$. In this paper, we are interested in the following two questions: 
\begin{quote}
\begin{itemize}
\item
(Q1):~When does system \eqref{A-1} - \eqref{A-1-1} exhibit emergent dynamics?
\vspace{0.2cm}
\item
(Q2):~What are the roles of cubic coupling terms involving with $\kappa_{1}$ and $\kappa_{2}$?
\end{itemize}
\end{quote}
Our main results of this paper are related to above two questions. First set of results is concerned with the Cauchy problem for a homogeneous ensemble with $A_i = A$ for all $i =1, \cdots, N$:
\begin{equation} \label{A-1-2}
\begin{cases}
\displaystyle {\dot T}_i =A T_i +\kappa_{1}(T_cT_i^* T_i -T_i T_c^* T_i)+\kappa_{2}(T_iT_i^* T_c-T_i T_c^* T_i), \quad t > 0, \\
\displaystyle T_i(0) = T_i^0, \quad \quad 1 \leq i \leq N.
\end{cases}
\end{equation}
Our first result is concerned with solution splitting property of \eqref{A-1-2}. More precisely, we consider associated linear free flow and nonlinear flow, respectively:
\begin{equation} \label{A-2}
\begin{cases} 
{\dot F} =A F, \\
{\dot N}_i = \kappa_{1}(N_c N_i^* N_i -N_i N_c^* N_i)+\kappa_{2}(N_i N_i^* N_c-N_i N_c^* N_i).
\end{cases}
\end{equation}
Then, once the rank-4 tensor $A$ satisfies a proposed criterion in \eqref{C-1-1}, a solution \eqref{A-1-2} can be expressed as the composition of associated linear and nonlinear flows \eqref{A-2} (solution splitting property, Theorem \ref{T3.1}):  
\begin{equation} \label{A-2-0-0}
L(t) = e^{t A} \quad \mbox{and} \quad T_j(t) = (L \circ N_j) T_j^0, \quad t \geq 0,~~i = 1, \cdots, N. 
\end{equation}
On the other hand, for $T_1, T_2 \in {\mathcal T}_2(\bbc;d_1 \times d_2)$, we introduce a Frobenius norm and inner product:
\[ \|T \|_F^2 := \sum_{\alpha, \beta} |[T]_{\alpha \beta}|^2 \quad \mbox{and} \quad \langle T_1, T_2 \rangle_F := \sum_{\alpha, \beta} \overline{[T_1]_{\alpha \beta}} [T_2]_{\alpha \beta}. \]
Here we used physicist's convention for inner product by taking the complex conjugate of the first argument.

Thus, thanks to the solution splitting property \eqref{A-2-0-0}, for a homogeneous ensemble, without loss of generality, we may assume $A \equiv 0$:
\begin{equation} \label{A-2-0}
\begin{cases} 
\displaystyle {\dot T}_i =\kappa_{1}(T_cT_i^* T_i -T_i T_c^* T_i)+\kappa_{2}(T_i T_i^* T_c-T_iT_c^* T_i), \quad t > 0, \\
\displaystyle T_i(0) = T_i^0 \in {\mathcal T}_2(\bbc; d_1 \times d_2), \quad \|T_i^0 \|_F = 1, \quad i = 1, \cdots, N.
\end{cases}
\end{equation}
Let $\{T_i \}$ be a solution to to \eqref{A-2-0}, we introduce the functional measuring the deviation from the centroid of configuration:
\[ {\mathcal V}[T(t)] :=\frac{1}{N}\sum_{k=1}^N||T_k(t)-T_c(t)||_F^2 = 1 - \|T_c(t)\|_F^2, \quad t \geq 0. \]
Here we used the conservation of norm (see Lemma \ref{L2.1}):
\[ \| T_i(t) \|_F = 1, \quad t \geq 0, \quad i = 1, \cdots, N. \]
In Theorem \ref{T4.1}, we show that there exists a positive value ${\mathcal V}_\infty$ such that 
\[  \lim_{t\rightarrow\infty} \frac{d}{dt} {\mathcal V}[T(t)] = 0 \quad \mbox{and} \quad \lim_{t \to \infty} {\mathcal V}[T(t)] = {\mathcal V}_\infty. \]
Now, the second set of results deal with the reduced model for \eqref{A-1} taking only the one of cubic interaction term:
\begin{equation} \label{A-3}
\begin{cases} 
\displaystyle {\dot T}_i = A_i T_i + \kappa_{1}(T_cT_i^* T_i -T_iT_c^* T_i), \quad t > 0, \\
\displaystyle T_i(0) = U_i^0 \Sigma V^*, \quad i = 1, \cdots, N.
\end{cases}
\end{equation}
We set 
\[ {\mathcal D}(A) := \max_{1 \leq i, j \leq N} \|A_i - A_j \|_{F}. \]
By the singular value decomposition(SVD) and conservation of rank of $T_i$, we can write $T_i$ as 
\[
T_i(t)=U_i(t)\Sigma V^*, \quad t > 0,~~i = 1, \cdots, N,
\] 
where $U_i = U_i(t)$ and $V$ are $d_1 \times d_1$ time-dependent and constant  $d_2 \times d_2$ unitary matrices, respectively, and $\Sigma$ is the $d_1\times d_2$ rectangular constant diagonal matrix with non-negative real numbers on the diagonal:
\[
[\Sigma_i]_{jk}=
\begin{cases}
\lambda_j\qquad\mbox{when $j=k\leq\min \{d_1, d_2 \}$,}\\
0\qquad \mbox{otherwise.}
\end{cases}
\]
(see Theorem \ref{T5.1}).  For the reduced model \eqref{A-3}, we will show that the Cauchy problem \eqref{A-3} is equivalent to the Lohe matrix model \cite{H-K-P-R} with a diagonal frustration  matrix $D$ (Theorem \ref{T5.1}).
\begin{align*}\label{A-4}
\begin{cases}
\dot{U}_i= B_iU_i + \kappa_{1}(U_cD-U_iD U_c^*U_i), \quad t > 0, \\
U_i(0)=U_i^0, \quad i = 1, \cdots, N.
\end{cases}
\end{align*}
Next, we consider the cases:
\[ \mbox{Either} \quad {\mathcal D}(B) = 0 \quad \mbox{or} \quad {\mathcal D}(B)  > 0. \]
For the former case with ${\mathcal D}(B) = 0$, we derive a differential inequality for $X = {\mathcal D}(U)^2$:
\[
-4\kappa_1 \mathcal{A}X+2\kappa_1 \mathcal{A}X^2\leq \frac{dX}{dt} \leq-4\kappa_1 \mathcal{B}X+2\kappa_1 \mathcal{A}X^2, \quad \mbox{a.e.}~t \in (0, \infty),
\]
where ${\mathcal A}$ and ${\mathcal B}$ a positive constant.
Thus, for a small initial data with ${\mathcal D}(U^0) \ll 1,$ one has exponential aggregation (See Theorem \ref{T6.1}):
\[  \mathcal{O}(1)e^{-2\kappa_1 \mathcal{A}t}\leq {\mathcal D}(U(t))\leq \mathcal{O}(1) e^{-2\kappa_1 \mathcal{B}t}, \quad \mbox{as $t \to \infty$}. \]
On the other hand, for a heterogeneous ensemble,  if the coupling strength and initial data satisfy
\[ \kappa_1 \gg {\mathcal D}(B),\qquad {\mathcal D}(U^0) \ll 1, \]
then, one has practical aggregation (Theorem \ref{T6.2}):
\[
\lim_{\kappa_1 \rightarrow\infty}\limsup_{t\rightarrow\infty} {\mathcal D}(U)=0.
\]

The rest of this paper is organized as follows. In Section \ref{sec:2}, we briefly introduce the Lohe tensor model which is a first-order aggregation model on the space of rank-m tensors with the same size, and discuss its reductions to the generalized Lohe sphere and Lohe matrix models. In Section \ref{sec:3}, we present a criterion for the solution splitting property of the generalized Lohe matrix model with the same $A_i$, and as an application of this criterion, we show that our general Lohe matrix model satisfies the criterion. In Section \ref{sec:4}, we study emergent dynamics for the homogeneous ensemble to the generalized Lohe matrix model and provide a decay estimate using the variance functional. In Section \ref{sec:5}, we present a reformulation of the reduced generalized Lohe matrix model which takes only one of cubic coupling to the Lohe matrix model with a diagonal frustration matrix, and provide a reformulation into the Lohe matrix model with a diagonal frustration matrix. In Section \ref{sec:6}, we present a practical aggregation for the reduced generalized Lohe matrix model using the diameter of state ensemble. Finally, Section \ref{sec:7} is devoted to a brief summary of our main results and some issues to be explored in a future study.  

\section{Preliminaries} \label{sec:2}
\setcounter{equation}{0}
In this section, we briefly introduce the Lohe tensor model \cite{H-P2}, and discuss relations between the generalized complex Lohe sphere model \cite{H-P1} and the Lohe matrix model \cite{Lo-0, Lo-1, Lo-2}. 

\subsection{The Lohe tensor model} \label{sec:2.1}
A tensor can be visualized as a $m$-dimensional array of complex numbers with $m$-indices. The rank of a tensor is the number of indices, say a rank-$m$ tensor of dimensions $d_1 \times \cdots \times d_m$ is an element of ${\mathbb C}^{d_1 \times \cdots \times d_m}$. Then, it is easy to see that scalars, vectors and matrices correspond to rank-0, 1 and 2 tensors, respectively.  Let $T$ be a rank-$m$ tensor with a size $d_1 \times \cdots \times d_m$.  Then, we denote $(\alpha_1, \cdots, \alpha_m)$-th component of the tensor $T$ by $[T]_{\alpha_1 \cdots \alpha_m}$, and we set $\overline{T}$ by the rank-$m$ tensor whose components are the complex conjugate of the elements in $T$:
\[ [\overline{T}]_{\alpha_1 \cdots \alpha_m} =\overline{[T]_{\alpha_1 \cdots \alpha_m}}. \]
Let ${\mathcal T}_m(\bbc; d_1 \times\cdots\times d_m)$ be the complex vector space consisting of all rank-$m$ tensors with size $d_1 \times\cdots\times d_m$ with standard addition and scalar multiplication. Then, the Kuramoto model, the Lohe sphere model and the Lohe matrix model can be regarded as aggregation models on ${\mathcal T}_0(\bbr; 0), {\mathcal T}_1(\bbr; d)$ and ${\mathcal T}_2(\bbc; d \times d)$, respectively. Let $\{T_j \}_{j=1}^{N}$ be the $N$-collection of rank-$m$ tensors in ${\mathcal T}_m(\bbc; d_1 \times\cdots\times d_m)$, and $A_j$ is the skew-hermitian rank-$2m$ tensors with size $(d_1 \times\cdots\times d_m) \times (d_1 \times \cdots\times d_m)$. For the simplicity of presentation, we introduce handy notation as follows: for $T \in {\mathcal T}_m(\bbc; d_1 \times \cdots\times d_m)$ and $A \in  {\mathcal T}_{2m}(\bbc; d_1 \times\cdots\times  d_m \times d_1 \times \cdots\times d_m)$, we set
\begin{align*}
\begin{aligned}
& [T]_{\alpha_{*}}:=[T]_{\alpha_{1}\alpha_{2}\cdots\alpha_{m}}, \quad [T]_{\alpha_{*0}}:=[T]_{\alpha_{10}\alpha_{20}\cdots\alpha_{m0}},  \quad  [T]_{\alpha_{*1}}:=[T]_{\alpha_{11}\alpha_{21}\cdots\alpha_{m1}}, \\
&  [T]_{\alpha_{*i_*}}:=[T]_{\alpha_{1i_1}\alpha_{2i_2}\cdots\alpha_{mi_m}}, \quad [T]_{\alpha_{*(1-i_*)}}:=[T]_{\alpha_{1(1-i_1)}\alpha_{2(1-i_2)}\cdots\alpha_{m(1-i_m)}}, \\
&  [A]_{\alpha_*\beta_*}:=[A]_{\alpha_{1}\alpha_{2}\cdots\alpha_{m}\beta_1\beta_2\cdots\beta_{m}}.
\end{aligned}
\end{align*}
Then, the Lohe tensor model \cite{H-P2} in component form can be written as follows:
\begin{equation}
\begin{cases} \label{B-0}
\displaystyle \dot{[T_j]}_{\alpha_{*0}} = [A_j]_{\alpha_{*0}\alpha_{*1}}[T_j]_{\alpha_{*1}} \\
\displaystyle \hspace{1cm} + \sum_{i_* \in \{0, 1\}^m}\kappa_{i_*} \Big([T_c]_{\alpha_{*i_*}}\bar{[T_j]}_{\alpha_{*1}}[T_j]_{\alpha_{*(1-i_*)}}-[T_j]_{\alpha_{*i_*}}\bar{[T_c]}_{\alpha_{*1}}[T_j]_{\alpha_{*(1-i_*)}} \Big), \\
\displaystyle  \bar{[A_j]}_{\alpha_{*0}\alpha_{*1}}=-[A_j]_{\alpha_{*1}\alpha_{*0}},
\end{cases}
\end{equation}
where $\kappa_{i_*}$'s are coupling strengths. \newline

The aforementioned models are generalization of the Kuramoto model \cite{A-R, B-C-M, C-H-J-K, C-S, D-X, D-B1, D-B0, D-B, H-K-R, H-L-X, Ku1, Ku2, M-S1, M-S2,M-S3, V-M1, V-M2} in the context of higher-dimensional sphere, unitary group and quaternions \cite{C-C-H, C-H1, C-H2, C-H3, C-H5, D-F-M-T, D-F-M, De, H-K-P-R, H-K, J-C, M-T-G,T-M, Zhu} and $L^2$-space \cite{C-H4}.
\subsection{From the Lohe tensor to the generalized Lohe matrix} \label{sec:2.2}
In this subsection, we discuss how the generalized Lohe matrix model can be derived from the Lohe tensor model as a special case. \newline

Consider the unitary group ${\mathbb U}(d)$:
\[ {\mathbb U}(d) := \{ U \in \bbc^{d \times d}~:~U^* U = U U^* = I_d \}, \]
where $I_d$ is the $d \times d$ identity matrix. Then, the Lohe matrix model in mean-field form becomes
\begin{equation} \label{B-1}
\mathrm{i}\dot{U}_jU_j^*=H_j+\frac{\mathrm{i}\kappa}{2}(U_cU_j^*-U_jU_c^*),
\end{equation}
where $U_c :=\frac{1}{N}\sum_{k=1}^NU_k$ is the average of $U_k$'s. Moreover, system \eqref{B-1} can be rewritten as follows:
\begin{align}
\begin{aligned} \label{B-2}
\dot{U}_j &=-{\mathrm i} H_j U_j+\frac{\kappa}{2}(U_cU_j^* U_j -U_jU_c^* U_j) \quad \mbox{or} \\
\dot{U}_j &=-{\mathrm i} H_j U_j+\frac{\kappa}{2}(U_j U^*_j U_c -U_jU_c^* U_j).
\end{aligned}
\end{align}
In literature \cite{B-C-S, De, H-X, H-K-R, H-R, Lo-0}, matrix-valued aggregation models including the Lohe matrix model \eqref{B-1} or \eqref{B-2} are defined on the space of square matrices. Recall that our purpose in this paper is to provide a first-order aggregation model on the space of non-square matrices.  Now, we consider the Lohe tensor model \eqref{B-0} with $m=2$:
\begin{align}
\begin{aligned} \label{B-3}
\dot{T}_j &= A_j T_j + \kappa_{00}(\mathrm{tr}(T_j^* T_j)T_c-\mathrm{tr}(T_c^* T_j)T_j) + + \kappa_{11}\mathrm{tr}(T_j^* T_c-T_c^* T_j)T_j \\
&\hspace{0.2cm} + \kappa_{10}(T_j T_j^* T_c-T_j T_c^* T_j) +  \kappa_{01}(T_cT_j^* T_j -T_jT_c^* T_j),
\end{aligned}
\end{align}
where the free flow term $A_jT_j$ is defined as a rank-2 tensor via tensor contraction:
\[
[A_jT_j]_{\alpha\beta}=[A_j]_{\alpha\beta\gamma\delta}[T_j]_{\gamma\delta}.
\]
Note that the first three terms in the R.H.S. of \eqref{B-3} are associated with the complex Lohe sphere model \cite{H-P1} on ${\mathcal T}_1(\bbc; d)$:
\[  \dot{T}_j= A_j T_j +\kappa_{0} (T_c\langle T_j, T_j \rangle-T_j \langle T_c. T_j \rangle )+\kappa_1(\langle{T_j, T_c}\rangle- \langle T_c, T_j\rangle) T_j. \]
where $\langle v, w \rangle = v^* w$ is the standard inner product in $\bbc^d$. In \eqref{B-3}, we set 
\[ \kappa_{00}=\kappa_{11}=0, \quad \kappa_1 = \kappa_{01}, \quad \kappa_2 = \kappa_{10}, \]
to obtain generalized Lohe matrix model \eqref{A-1}. 
and study the conservation of Frobenius norm for $T_j$.
\begin{lemma} \label{L2.1}
Let $\{T_j \}_{i=1}^N$ be a solution to \eqref{A-1} with initial data $\{T_j^0 \}$. Then the Frobenius norm of $T_j$ is conserved along the flow \eqref{A-1}:
\[ \|T_j(t) \|_F = \|T_j^0 \|_F, \quad t > 0,~~ j = 1, \cdots, N. \]
\end{lemma}
\begin{proof} We use $\eqref{A-1}$ to see
\begin{align}
\begin{aligned} \label{B-4-1}
& \frac{d}{dt}||T_j||_F^2  =\langle{\dot{T_j}, T_j}\rangle_F+\langle{T_j, \dot{T_j}}\rangle_F \\
& \hspace{0.2cm}  = \Big \langle{T_j, A_jT_j +\kappa_{1}(T_cT_j^* T_j -T_j T_c^* T_j)+\kappa_{2}(T_jT_j^* T_c-T_j T_c^* T_j)} \Big \rangle+\mbox{(c.c.)}\\
&\hspace{0.2cm}  =\langle{T_j, A_j T_j}\rangle+\kappa_{1}\langle{T_j, T_cT_j^*T_j-T_jT_c^*T_j}\rangle+\kappa_{2}\langle{T_j, T_j T_j^*T_c-T_j T_c^*T_j}\rangle+\mbox{(c.c.)}\\
&\hspace{0.2cm}  =\bar{[T_j]}_{\alpha\beta}[A_jT_j]_{\alpha\beta}+\kappa_{1}(\bar{[T_j]}_{\alpha\beta}[T_cT_j^*T_j-T_j T_c^*T_j]_{\alpha\beta}) +\kappa_{2}\bar{[T_j]}_{\alpha\beta}[T_j T_j^*T_c-T_j T_c^*T_j]_{\alpha\beta}+\mbox{(c.c.)}\\
&=\bar{[T_j]}_{\alpha\beta}[A_j]_{\alpha\beta\gamma\delta}[T_j]_{\gamma\delta}+ \kappa_{1}\mathrm{tr}(T_j^*T_cT_j^*T_j -T_j^*T_j T_c^*T_j)  +\kappa_{2}\mathrm{tr}(T_j^*T_j T_j^*T_c-T_j^*T_jT_c^*T_j)+\mbox{(c.c.)}\\
&\hspace{0.2cm}  = {\mathcal I}_{11} +  {\mathcal I}_{12}  +  {\mathcal I}_{13} + \overline{{\mathcal I}_{11}} +  \overline{{\mathcal I}_{12}}  +  \overline{{\mathcal I}_{13}},
\end{aligned}
\end{align}
where $c.c$ denotes the complex conjugate of the proceeding terms, and we used the Einstein summation rule for repeated dummy variables.  \newline

\noindent $\bullet$ (Estimate of $ {\mathcal I}_{11} + \overline{{\mathcal I}_{11}}$): We use the change of dummy variables and \eqref{A-0-1}:
\[  [{\bar A}]_{\alpha \beta \gamma \delta} = -[A]_{\gamma \delta \alpha \beta} \quad \mbox{and} \quad (\alpha, \beta) \longleftrightarrow (\gamma, \delta)   \]
to see
\begin{align}
\begin{aligned} \label{B-5-1}
{\mathcal I}_{11} + \overline{{\mathcal I}_{11}} &= \bar{[T_j]}_{\alpha\beta}[A_j]_{\alpha\beta\gamma\delta}[T_j]_{\gamma\delta}  +  [T_j]_{\alpha\beta}[{\bar A}_j]_{\alpha\beta\gamma\delta}[{\bar T}_j]_{\gamma\delta}  \\
&= \bar{[T_j]}_{\alpha\beta}[A_j]_{\alpha\beta\gamma\delta}[T_j]_{\gamma\delta} -  [T_j]_{\alpha\beta}[A_j]_{\gamma\delta \alpha\beta}[{\bar T}_j]_{\gamma\delta} \\
&= \bar{[T_j]}_{\alpha\beta}[A_j]_{\alpha\beta\gamma\delta}[T_j]_{\gamma\delta} -  [T_j]_{\gamma\delta}[A_j]_{\alpha\beta \gamma \delta}[{\bar T}_j]_{\alpha\beta} \\
&=0. 
\end{aligned}
\end{align}
\noindent $\bullet$ (Estimate of $ {\mathcal I}_{12} + \overline{{\mathcal I}_{12}}$): By direct calculation, one has
\begin{align}
\begin{aligned}  \label{B-5-2}
{\mathcal I}_{12} + \overline{{\mathcal I}_{12}} =\kappa_{1} \Big[  \mathrm{tr}(T_j^*T_cT_j^*T_j -T_j^*T_j T_c^*T_j)  +  
\mathrm{tr}(T_j^*T_j T_c^*T_j -T_j^*T_cT_j^*T_j )  \Big ]=0.
\end{aligned}
\end{align}
\noindent $\bullet$ (Estimate of $ {\mathcal I}_{13} + \overline{{\mathcal I}_{13}}$): Similarly, one has
\begin{align}
\begin{aligned} \label{B-5-3}
 {\mathcal I}_{13} + \overline{{\mathcal I}_{13}} = \kappa_{2} \Big[ \mathrm{tr}(T_j^*T_j T_j^*T_c-T_j^*T_jT_c^*T_j) + \mathrm{tr}(T_j^*T_jT_c^*T_j -T_j^*T_j T_j^*T_c)  \Big] = 0.
\end{aligned}
\end{align}
In \eqref{B-4-1}, we combine all the estimates \eqref{B-5-1}, \eqref{B-5-2} and \eqref{B-5-3} to get the desired estimate:
\[  \frac{d}{dt}||T_j||_F^2 = 0.    \]
\end{proof}
Next, we study evolution of the hermitian conjugate of $T_j$. For this, we set
\begin{equation} \label{B-5}
 S_j:=T_j^* \quad j = 1, \cdots, N, \quad \mbox{and} \quad S_c:=T_c^*. 
\end{equation} 
\begin{lemma} \label{L2.2}
Let $\{T_j \}_{i=1}^N$ be a solution to \eqref{A-1}. Then $S_j$ in \eqref{B-5} satisfies
\begin{align}
\begin{cases}  \label{B-6}
\displaystyle {\dot S}_j=B_j S_j +\kappa_{2}(S_cS_j^* S_j - S_jS_c^* S_j)+\kappa_{1}(S_j S_j^* S_c-S_j S_c^* S_j), \quad t > 0,\\
\displaystyle S_j(0)=S_j^0\quad\mbox{for all $j=1, 2, \cdots, N$},
\end{cases}
\end{align}
where the rank-4 tensor $B_j$ is given as follows.
\[
[B_j]_{\alpha\beta\gamma\delta}=\overline{[A_j]_{\beta\alpha\delta\gamma}}, \quad (A_jT_j)^*=B_j T_j^*. 
\]
\end{lemma}
\begin{proof} We take the hermitian conjugation of \eqref{A-1} and use $\kappa_{01},~\kappa_{10} \in \bbr$ to get 
\[ {\dot T}_j^*=(A_j T_j)^*+\kappa_{1}(T_j^* T_j T_c^*-T_j^* T_c T_j^*)+\kappa_{2}(T_c^* T_j T_j^*- T_j^* T_c T_j^*).
\]
Note that 
\[
[(A_j T_j)^*]_{\alpha\beta}=\overline{[A_j T_j]_{\beta\alpha}}=\overline{[A_j]_{\beta\alpha\delta\gamma}[T_j]_{\delta\gamma}}=\overline{[A_j]_{\beta\alpha\delta\gamma}}[T_j^*]_{\gamma\delta}.
\]
Then, $S_j$ in \eqref{B-5} implies  
\begin{align*}\label{B-7}
\begin{cases}
{\dot S}_j=B_j S_j+\kappa_{2}(S_cS_j^* S_j - S_j S_c^* S_j)+\kappa_{1}(S_j S_j^* S_c-S_j S_c^* S_j),\\
S_j(0)=S_j^0\quad\mbox{for all $j=1, 2, \cdots, N$}.
\end{cases}
\end{align*}
\end{proof}
\begin{remark}
If we compare two systems \eqref{A-1} and \eqref{B-6}, it is easy to see that two coupling strengths $\kappa_{1}$ and $\kappa_{2}$ were interchanged.
\end{remark}

\subsection{From generalized Lohe matrix to Lohe sphere} \label{sec:2.3}
First, we recall the complex Lohe sphere model for rank-1 tensor $T_j = z_j$:
\begin{equation} \label{B-7-1}
\dot{z}_j= \Omega_j z_j +\kappa_{1} (z_c\langle z_j, z_j \rangle-z_j \langle z_c. z_j \rangle )+\kappa_2(\langle{z_j, z_c}\rangle- \langle z_c z_j\rangle) z_j,
\end{equation}
where $\langle \cdot, \cdot \rangle$ is the standard inner product in $\bbc^d$. Below, we discuss how to reduce \eqref{B-7-1} from \eqref{A-1}. \newline

Let ${\tilde T}_j$ be a column complex vector with size $d \times 1$ for all $j=1, 2, \cdots, N$, i.e.,
\[
d_1 = d, \quad d_2 = 1.
\]
Then, one has 
\begin{equation} \label{B-7-2}
T_j^* T_j = \langle T_j, T_j \rangle, \quad   T_c^* T_j = \langle T_c, T_j \rangle, \quad T_j^* T_c  = \langle T_j, T_c \rangle, \quad T_c^* T_j = \langle T_c, T_j \rangle.
\end{equation}
Recall the generalized Lohe matrix model:
\begin{equation} \label{B-7-3}
{\dot T}_j =A_j T_j +\kappa_{1}(T_cT_j^* T_j -T_j T_c^* T_j)+\kappa_{2}(T_j T_j^* T_c-T_j T_c^* T_j).
\end{equation}
We rewrite \eqref{B-7-3} using \eqref{B-7-2} to see the complex Lohe sphere model:
\[
{\dot T}_j =A_j T_j+\kappa_1 (T_c\langle T_j, T_j \rangle -T_j \langle T_c, T_j \rangle)+\kappa_{2}(\langle T_j, T_c \rangle-
 \langle T_c, T_j \rangle) T_j,
\]
where $\langle A, B \rangle$ is an inner product of two complex vectors using the Einstein summation rule:
\[
\langle A, B \rangle := [\bar{A}]_{\alpha}[B]_{\alpha}.
\]
 \section{Solution splitting property}  \label{sec:3}
\setcounter{equation}{0} 
In this section, we present a solution splitting property of the generalized Lohe matrix model with an identical free flow.  \newline

Consider the Cauchy problem for the generalized Lohe matrix model for a homogeneous ensemble:
\begin{equation}\label{C-1}
\begin{cases}
{\dot T}_j=AT_j+\kappa_{1}(T_cT_i^* T_i-T_iT_c^* T_i)+\kappa_{2}(T_iT_i^* T_c-T_iT_c^* T_i), \quad t > 0, \\
T_j(0) = T_j^0, \quad j = 1, \cdots, N.
\end{cases}
\end{equation}
Before we present solution splitting property, we begin with definition of exponential of rank-4 tensors in analogy with matrix exponential as follows.
\begin{definition} \label{D3.1}
Let $X$ be a rank-4 tensor. Then, the $n$-th power of $X$ and exponential of $X$ are defined as follows.
\begin{align*}
\begin{aligned}
& [X^n]_{\alpha\beta\gamma\delta} :=[X]_{\alpha\beta\alpha_1\beta_1}[X]_{\alpha_1\beta_1\alpha_2\beta_2}\cdots[X]_{\alpha_{n-1}\beta_{n-1}\gamma\delta}, \\
& [e^{X}]_{\alpha\beta\gamma\delta} :=\sum_{n=0}^\infty \frac{1}{n!}[X^n]_{\alpha\beta\gamma\delta},
\end{aligned}
\end{align*}
where we used the Einstein summation rule for repeated indices.
\end{definition}
Next, we study the solution splitting property of \eqref{C-1}. For this, we set
\[
N_j:=e^{-At}T_j, \quad j = 1, \cdots, N.
\]
Below, we will show that $S_j$ satisfies system \eqref{C-1} with $A \equiv 0$. 
\begin{theorem} \label{T3.1}
Suppose that rank-4 tensor $A$ satisfies the following property:
\begin{equation} \label{C-1-1}
[e^{-At}]_{\alpha\beta\gamma\delta}[e^{At}]_{\gamma\epsilon\alpha_1\beta_1}[e^{-At}]_{\alpha_2\beta_2\psi\epsilon}[e^{At}]_{\psi\delta\alpha_3\beta_3}=\delta_{\alpha_1\alpha}\delta_{\beta_3\beta}\delta_{\beta_1\beta_2}\delta_{\alpha_2\alpha_3},
\end{equation}
and let $\{ T_j \}$ be a solution to \eqref{C-1}. Then, one has 
\[ T_j = e^{At} N_j, \quad j = 1, \cdots, N, \]
where $N_j$ satisfies 
\[ {\dot N}_j=\kappa_{1}(N_cN_j^* N_j-N_jN_c^* N_j)+\kappa_{2}(N_j N_j^* N_c-N_j N_c^* N_j). \]
\end{theorem}
\begin{proof}
We substitute the ansatz $T_j=e^{At}N_j$ into \eqref{C-1} to obtain
\begin{align*}
\begin{aligned}
& Ae^{At} N_j+e^{At}\dot{N}_j =Ae^{At} N_j+\kappa_{1}([e^{At}N_c][e^{At}N_j]^*[e^{At}N_j]-[e^{At}N_j][e^{At}N_c]^*[e^{At}N_j])\\
& \hspace{2cm} +\kappa_{2}([e^{At} N_j][e^{At} N_j]^*[e^{At} N_c]-[e^{At}N_j][e^{At}N_c]^*[e^{At}N_j]).
\end{aligned}
\end{align*}
After further simplification, one has 
\begin{align*}
\dot{N}_j&=\kappa_{1}e^{-At}([e^{At}N_c][e^{At}N_j]^*[e^{At}N_j]-[e^{At}N_j][e^{At}N_c]^*[e^{At}N_j])\\
& \hspace{0.5cm} +\kappa_{2}e^{-At}([e^{At}N_j][e^{At}N_j]^*[e^{At}N_c]-[e^{At}N_j][e^{At}N_c]^*[e^{At}N_j]).
\end{align*}
If we express above relation in component form, we can obtain
\begin{align*}
\frac{d}{dt} [N_j]_{\alpha\beta}&=\kappa_{1}[e^{-At}]_{\alpha\beta\gamma\delta}[([e^{At} N_c][e^{At}N_j]^*[e^{At}N_j]-[e^{At}N_j][e^{At}N_c]^*[e^{At}N_j])]_{\gamma\delta}\\
&+\kappa_{2}[e^{-At}]_{\alpha\beta\gamma\delta}[([e^{At}N_j][e^{At}N_j]^*[e^{At}N_c]-[e^{At}N_j][e^{At}N_c]^*[e^{At}N_j])]_{\gamma\delta}.
\end{align*}
Next, we simplify the term $[[e^{At}N_c][e^{At}N_j]^*[e^{At}N_j]]_{\gamma\delta}$ as follows.
\begin{align*}
\begin{aligned}
&[[e^{At}N_c][e^{At}N_j]^*[e^{At}N_j]]_{\gamma\delta} \\
& \hspace{0.5cm} =[e^{At}N_c]_{\gamma\epsilon}[e^{At}N_j]^*_{\epsilon\psi}[e^{At}N_j]_{\psi\delta}\\
& \hspace{0.5cm}  =[e^{At}]_{\gamma\epsilon\alpha_1\beta_1}[N_c]_{\alpha_1\beta_1}\overline{[e^{At}]_{\psi\epsilon\alpha_2\beta_2}[N_j]_{\alpha_2\beta_2}}[e^{At}]_{\psi\delta\alpha_3\beta_3}[N_j]_{\alpha_3\beta_3}\\
& \hspace{0.5cm}  =[e^{At}]_{\gamma\epsilon\alpha_1\beta_1}\overline{[e^{At}]_{\psi\epsilon\alpha_2\beta_2}}[e^{At}]_{\psi\delta\alpha_3\beta_3}[N_c]_{\alpha_1\beta_1}\overline{[N_j]_{\alpha_2\beta_2}}[N_j]_{\alpha_3\beta_3}\\
& \hspace{0.5cm}  =[e^{At}]_{\gamma\epsilon\alpha_1\beta_1}[e^{-At}]_{\alpha_2\beta_2\psi\epsilon}[e^{At}]_{\psi\delta\alpha_3\beta_3}[N_c]_{\alpha_1\beta_1}\overline{[N_j]_{\alpha_2\beta_2}}[N_j]_{\alpha_3\beta_3}.
\end{aligned}
\end{align*}
Then we can obtain following from similar operation to other terms.
\begin{align}
\begin{aligned}\label{C-2}
\frac{d}{dt} [N_j]_{\alpha\beta}&=[e^{-At}]_{\alpha\beta\gamma\delta}[e^{At}]_{\gamma\epsilon\alpha_1\beta_1}[e^{-At}]_{\alpha_2\beta_2\psi\epsilon}[e^{At}]_{\psi\delta\alpha_3\beta_3} \\
&\hspace{1cm}\times (\kappa_{1}([N_c]_{\alpha_1\beta_1}\overline{[N_j]_{\alpha_2\beta_2}}[N_j]_{\alpha_3\beta_3}-[N_j]_{\alpha_1\beta_1}\overline{[N_c]_{\alpha_2\beta_2}}[N_j]_{\alpha_3\beta_3})\\
&\hspace{2cm}+\kappa_{2}([N_j]_{\alpha_1\beta_1}\overline{[N_j]_{\alpha_2\beta_2}}[N_c]_{\alpha_3\beta_3}-[N_j]_{\alpha_1\beta_1}\overline{[N_c]_{\alpha_2\beta_2}}[N_j]_{\alpha_3\beta_3})).
\end{aligned}
\end{align}
Note that we want to show the solution splitting property to yield the following form:
\[
\frac{dN_j}{dt} =\kappa_{1}(N_cN_j^* N_j-N_jN_c^* N_j)+\kappa_{2}(N_j N_j^* N_c-N_jN_c^* N_j).
\]
To compare components, we can obtain the following condition:
\begin{align}\label{C-3}
[e^{-At}]_{\alpha\beta\gamma\delta}[e^{At}]_{\gamma\epsilon\alpha_1\beta_1}[e^{-At}]_{\alpha_2\beta_2\psi\epsilon}[e^{At}]_{\psi\delta\alpha_3\beta_3}=\delta_{\alpha_1\alpha}\delta_{\beta_3\beta}\delta_{\beta_1\beta_2}\delta_{\alpha_2\alpha_3}.
\end{align}
If we substitute the relation \eqref{C-3} into \eqref{C-2} to see
\begin{align*}
\frac{d}{dt} [N_j]_{\alpha\beta}&=[e^{-At}]_{\alpha\beta\gamma\delta}[e^{At}]_{\gamma\epsilon\alpha_1\beta_1}[e^{-At}]_{\alpha_2\beta_2\psi\epsilon}[e^{At}]_{\psi\delta\alpha_3\beta_3} \\
&\times (\kappa_{1}([N_c]_{\alpha_1\beta_1}\overline{[N_j]_{\alpha_2\beta_2}}[N_j]_{\alpha_3\beta_3}-[N_j]_{\alpha_1\beta_1}\overline{[N_c]_{\alpha_2\beta_2}}[N_j]_{\alpha_3\beta_3})\\
&+\kappa_{2}([N_j]_{\alpha_1\beta_1}\overline{[N_j]_{\alpha_2\beta_2}}[N_c]_{\alpha_3\beta_3}-[N_j]_{\alpha_1\beta_1}\overline{[N_c]_{\alpha_2\beta_2}}[N_j]_{\alpha_3\beta_3}))\\
&=\delta_{\alpha_1\alpha}\delta_{\beta_3\beta}\delta_{\beta_1\beta_2}\delta_{\alpha_2\alpha_3}(\kappa_{1}([N_c]_{\alpha_1\beta_1}\overline{[N_j]_{\alpha_2\beta_2}}[N_j]_{\alpha_3\beta_3}-[N_j]_{\alpha_1\beta_1}\overline{[N_c]_{\alpha_2\beta_2}}[N_j]_{\alpha_3\beta_3})\\
&+\kappa_{2}([N_j]_{\alpha_1\beta_1}\overline{[N_j]_{\alpha_2\beta_2}}[N_c]_{\alpha_3\beta_3}-[N_j]_{\alpha_1\beta_1}\overline{[N_c]_{\alpha_2\beta_2}}[N_j]_{\alpha_3\beta_3}))\\
&=\kappa_{1}([N_c]_{\alpha\beta_1}\overline{[N_j]_{\alpha_3\beta_1}}[N_j]_{\alpha_3\beta}-[N_j]_{\alpha\beta_1}\overline{[N_c]_{\alpha_3\beta_1}}[N_j]_{\alpha_3\beta})\\
&+\kappa_{2}([N_j]_{\alpha\beta_1}\overline{[N_j]_{\alpha_3\beta_1}}[N_c]_{\alpha_3\beta}-[N_j]_{\alpha\beta_1}\overline{[N_c]_{\alpha_3\beta_1}}[N_j]_{\alpha_3\beta})\\
&=[\kappa_{1}(N_cN_j^*N_j-N_jN_c^* N_j)+\kappa_{2}(N_jN_j^* N_c-N_jN_c^* N_j)]_{\alpha\beta}.
\end{align*}
Thus, one has 
\[
\frac{dN_j}{dt} =\kappa_{1}(N_c N_j^* N_j-N_jN_c^* N_j)+\kappa_{2}(N_jN_j^* N_c-N_jN_c^* N_j).
\]
\end{proof}

\subsection{Examples} In this subsection, we present two examples satisfying the condition \ref{C-1-1} in Theorem \ref{T3.1}. 
\subsubsection{Example A} Consider the Lohe matrix model:
\begin{equation} \label{C-4}
\mathrm{i}\dot{U}_jU_j^*=H_j+\frac{\mathrm{i}\kappa}{2N}(U_cU_j^*-U_jU_c^*),
\end{equation}
where $H_j$ is the hermitian matrix with $H_j^* = H_j.$
System \eqref{C-4} can be rewritten as follows.
\begin{equation*} \label{C-5}
\dot{U}_j=-\mathrm{i}H_jU_j+\frac{\kappa}{N}(U_c-U_jU_c^*U_j).
\end{equation*}
Then we have
\[
A_jU_j=-\mathrm{i}H_jU_j \quad \mbox{i.e.,} \quad [A_j]_{\alpha\beta\gamma\delta}[U_j]_{\gamma\delta}=-\mathrm{i}[H_j]_{\alpha\gamma}[U_j]_{\gamma\beta}.
\]
This yields
\begin{equation} \label{C-6}
[A_j]_{\alpha\beta\gamma\delta}=-\mathrm{i}[H_j]_{\alpha\gamma}\delta_{\delta\beta}.
\end{equation}
\begin{proposition} \label{P3.1}
The following estimates hold.
\[ [\bar{A_j}]_{\gamma\delta\alpha\beta} = -[A_j]_{\alpha\beta\gamma\delta}, \quad  [e^{-At}]_{\alpha\beta\gamma\delta}[e^{At}]_{\gamma\epsilon\alpha_1\beta_1}[e^{-At}]_{\alpha_2\beta_2\psi\epsilon}[e^{At}]_{\psi\delta\alpha_3\beta_3} = \delta_{\alpha_1\alpha}\delta_{\beta_3\beta}\delta_{\beta_1\beta_2}\delta_{\alpha_2\alpha_3}.
\]
\end{proposition}
\begin{proof}
(i) By direct calculation, one has
\[
[\bar{A_j}]_{\gamma\delta\alpha\beta}=\mathrm{i}\bar{[H_j]}_{\gamma\alpha}\delta_{\beta\delta}=\mathrm{i}[H_j]_{\alpha\gamma}\delta_{\delta\beta}=-[A_j]_{\alpha\beta\gamma\delta}.
\]
(ii) Next, we  show that $A_j = A$ defined by \eqref{C-6} satisfies the relations \eqref{C-1-1}. First, note that 
\[
[e^{At}]_{\alpha\beta\gamma\delta}=[e^{-\mathrm{i}Ht}]_{\alpha\gamma}\delta_{\delta\beta}, \quad [e^{-At}]_{\alpha\beta\gamma\delta}=[e^{\mathrm{i}Ht}]_{\alpha\gamma}\delta_{\delta\beta}.
\]
This yields
\begin{align*}
\begin{aligned}
&[e^{-At}]_{\alpha\beta\gamma\delta}[e^{At}]_{\gamma\epsilon\alpha_1\beta_1}[e^{-At}]_{\alpha_2\beta_2\psi\epsilon}[e^{At}]_{\psi\delta\alpha_3\beta_3}\\
& \hspace{0.5cm} =[e^{\mathrm{i}Ht}]_{\alpha\gamma}\delta_{\beta\delta}[e^{-\mathrm{i}Ht}]_{\gamma\alpha_1}\delta_{\beta_1\epsilon}[e^{\mathrm{i}Ht}]_{\alpha_2\psi}\delta_{\beta_2\epsilon}[e^{-\mathrm{i}Ht}]_{\psi\alpha_3}\delta_{\beta_3\delta}\\
& \hspace{0.5cm}  =[e^{\mathrm{i}Ht}]_{\alpha\gamma}[e^{-\mathrm{i}Ht}]_{\gamma\alpha_1}[e^{\mathrm{i}Ht}]_{\alpha_2\psi}[e^{-\mathrm{i}Ht}]_{\psi\alpha_3}\delta_{\beta\delta}\delta_{\beta_1\epsilon}\delta_{\beta_2\epsilon}\delta_{\beta_3\delta}\\
&  \hspace{0.5cm} =\delta_{\alpha\alpha_1}\delta_{\alpha_2\alpha_3}\delta_{\beta\delta}\delta_{\beta_1\epsilon}\delta_{\beta_2\epsilon}\delta_{\beta_3\delta}\\
& \hspace{0.5cm}  =\delta_{\alpha\alpha_1}\delta_{\alpha_2\alpha_3}\delta_{\beta\beta_3}\delta_{\beta_1\beta_2}\\
& \hspace{0.5cm}  =\delta_{\alpha_1\alpha}\delta_{\beta_3\beta}\delta_{\beta_1\beta_2}\delta_{\alpha_2\alpha_3}.
\end{aligned}
\end{align*}
Therefore, we checked that the relation \eqref{C-6} satisfies the condition  \eqref{C-1-1} in Theorem \ref{T3.1}.
\end{proof}

\subsubsection{Example B} Consider the rank-4 tensor $A$ satisfying the relation:
\begin{equation} \label{C-8}
[A]_{\alpha\beta\gamma\delta} :=[B]_{\alpha\gamma}\delta_{\beta\delta}+[C]_{\beta\delta}\delta_{\alpha\gamma}
\end{equation}
with skew-hermitian matrices $B$ and $C$ are rank-2 tensors. Then, it is easy to see that 
\[ -\overline{[A]_{\gamma\delta\alpha\beta}} =-(\overline{[B]_{\gamma\alpha}}\delta_{\delta\beta}+[C]_{\delta\beta}\delta_{\gamma\alpha})=[B]_{\alpha\gamma}\delta_{\beta\delta}+[C]_{\beta\delta}\delta_{\alpha\gamma}=[A]_{\alpha\beta\gamma\delta}. \]
Now we want to show that $A$ satisfies the condition \eqref{C-1-1}. We set
\[
[X]_{\alpha\beta\gamma\delta}=[B]_{\alpha\gamma}\delta_{\beta\delta},\qquad [Y]_{\alpha\beta\gamma\delta}=[C]_{\beta\delta}\delta_{\alpha\gamma}.
\]
Then, it follows from \eqref{C-8} that 
\[
A=X+Y.
\]
For rank-4 tensors $X$ and $Y$, we define their product $XY$ as
\[
[XY]_{\alpha\beta\gamma\delta}=[X]_{\alpha\beta\eta\psi}[Y]_{\eta\psi\gamma \delta}.
\]
\begin{proposition} \label{P3.2} The following estimates hold.
\[ XY = YX, \quad [e^{-At}]_{\alpha\beta\gamma\delta}[e^{At}]_{\gamma\epsilon\alpha_1\beta_1}[e^{-At}]_{\alpha_2\beta_2\psi\epsilon}[e^{At}]_{\psi\delta\alpha_3\beta_3} = \delta_{\alpha_1\alpha}\delta_{\beta_3\beta}\delta_{\beta_1\beta_2}\delta_{\alpha_2\alpha_3}.
\]
\end{proposition}
\begin{proof}
\noindent (i)~By direct calculation, one has
\[
[XY]_{\alpha\beta\gamma\delta}=[X]_{\alpha\beta\eta\psi}[Y]_{\eta\psi\gamma\delta}=[X]_{\alpha\eta}\delta_{\beta\psi}[Y]_{\psi\delta}\delta_{\eta\gamma}=[X]_{\alpha\gamma}[Y]_{\beta\delta},
\]
\[
[YX]_{\alpha\beta\gamma\delta}=[Y]_{\alpha\beta\eta\psi}[X]_{\eta\psi\gamma\delta}=[Y]_{\beta\psi}\delta_{\alpha\eta}[X]_{\eta\gamma}\delta_{\psi\delta}=[X]_{\alpha\gamma}[Y]_{\beta\delta}.
\]
By tensor contraction, we have

\[
XY=YX.
\]
\noindent (ii)~Note that 
\begin{align}
\begin{aligned} \label{C-9} 
[e^{At}]_{\alpha\beta\gamma\delta}  &= \sum_{n=0}^\infty \frac{1}{n!}[A^n]_{\alpha\beta\gamma\delta}t^n=\sum_{n=0}^\infty \frac{1}{n!}[(X+Y)^n]_{\alpha\beta\gamma\delta}t^n=[e^{Xt}e^{Yt}]_{\alpha\beta\gamma\delta} \\
&= [e^{Xt}]_{\alpha\beta\eta\psi}[e^{Yt}]_{\eta\psi\gamma\delta}.
\end{aligned}
\end{align}
On the other hand, note that 
\begin{align}
\begin{aligned} \label{C-10}
[X^n]_{\alpha\beta\gamma\delta} &=[X]_{\alpha\beta\alpha_1\beta_1}[X]_{\alpha_1\beta_1\alpha_2\beta_2}\cdots[X]_{\alpha_{n-1}\beta_{n-1}\gamma\delta}\\
&=[B]_{\alpha\alpha_1}[B]_{\alpha_1\alpha_2}\cdots[B]_{\alpha_{n-1}\gamma}\delta_{\beta\beta_1}\delta_{\beta_1\beta_2}\cdots\delta_{\beta_{n-1}\delta}\\
&=[B^n]_{\alpha\gamma}\delta_{\beta\delta}
\end{aligned}
\end{align}
and
\begin{align}
\begin{aligned} \label{C-11}
[Y^n]_{\alpha\beta\gamma\delta}&=[Y]_{\alpha\beta\alpha_1\beta_1}[Y]_{\alpha_1\beta_1\alpha_2\beta_2}\cdots[Y]_{\alpha_{n-1}\beta_{n-1}\gamma\delta}\\
&=[C]_{\beta\beta_1}[C]_{\beta_1\beta_2}\cdots[C]_{\beta_{n-1}\delta}\delta_{\alpha\alpha_1}\delta_{\alpha_1\alpha_2}\cdots\delta_{\alpha_{n-1}\gamma}\\
&=[C^n]_{\beta\delta}\delta_{\alpha\gamma}.
\end{aligned}
\end{align}
Thus, the factors $e^{Xt}$ and $e^{Yt}$ can be estimated as follows.
\begin{align}
\begin{aligned} \label{C-12}
& [e^{Xt}]_{\alpha\beta\gamma\delta}=\sum_{n=0}^\infty\frac{1}{n!}[X^n]_{\alpha\beta\gamma\delta}t^n=\sum_{n=0}^\infty\frac{1}{n!}[B]_{\alpha\gamma}t^n\delta_{\beta\delta}=[e^{Bt}]_{\alpha\gamma}\delta_{\beta\delta}, \\
& [e^{Yt}]_{\alpha\beta\gamma\delta}=\sum_{n=0}^\infty\frac{1}{n!}[Y^n]_{\alpha\beta\gamma\delta}t^n=\sum_{n=0}^\infty\frac{1}{n!}[C]_{\beta\delta}t^n\delta_{\alpha\gamma}=[e^{Ct}]_{\beta\delta}\delta_{\alpha\gamma}.
\end{aligned}
\end{align}
In \eqref{C-9}, we use \eqref{C-10}, \eqref{C-11} and \eqref{C-12} to get
\begin{equation} \label{C-13}
[e^{At}]_{\alpha\beta\gamma\delta}=[e^{Xt}]_{\alpha\beta\eta\psi}[e^{Yt}]_{\eta\psi\gamma\delta}=[e^{Bt}]_{\alpha\eta}\delta_{\beta\psi}[e^{Ct}]_{\psi\delta}\delta_{\eta\gamma}=[e^{Bt}]_{\alpha\gamma}[e^{Ct}]_{\beta\delta}.
\end{equation}
Next, we substitute $-s\rightarrow t$ into \eqref{C-13} to get
\begin{equation} \label{C-14}
[e^{-As}]_{\alpha\beta\gamma\delta}=[e^{-Bs}]_{\alpha\gamma}[e^{-Cs}]_{\beta\delta}.
\end{equation}
Now we use the relation \eqref{C-14} to obtain
\begin{align*}
\begin{aligned}
&[e^{-At}]_{\alpha\beta\gamma\delta}[e^{At}]_{\gamma\epsilon\alpha_1\beta_1}[e^{-At}]_{\alpha_2\beta_2\psi\epsilon}[e^{At}]_{\psi\delta\alpha_3\beta_3}\\
& \hspace{0.5cm} =[e^{-Bt}]_{\alpha\gamma}[e^{-Ct}]_{\beta\delta}[e^{Bt}]_{\gamma\alpha_1}[e^{Ct}]_{\epsilon\beta_1}[e^{-Bt}]_{\alpha_2\psi}[e^{-Ct}]_{\beta_2\epsilon}[e^{Bt}]_{\psi\alpha_3}[e^{Ct}]_{\delta\beta_3}\\
& \hspace{0.5cm} =[e^{-Bt}]_{\alpha\gamma}[e^{Bt}]_{\gamma\alpha_1}[e^{-Bt}]_{\alpha_2\psi}[e^{Bt}]_{\psi\alpha_3}[e^{-Ct}]_{\beta\delta}[e^{Ct}]_{\epsilon\beta_1}[e^{-Ct}]_{\beta_2\epsilon}[e^{Ct}]_{\delta\beta_3}\\
& \hspace{0.5cm} =\delta_{\alpha\alpha_1}\delta_{\alpha_2\alpha_3}\delta_{\beta\beta_3}\delta_{\beta_2\beta_1}=\delta_{\alpha_1\alpha}\delta_{\beta_3\beta}\delta_{\beta_1\beta_2}\delta_{\alpha_2\alpha_3}.
\end{aligned}
\end{align*}
This yields \eqref{C-1-1}. 
\end{proof}
Note that 
\[
[AT_i]_{\alpha\beta}=[A]_{\alpha\beta\gamma\delta}[T_i]_{\gamma\delta}=[B]_{\alpha\gamma}\delta_{\beta\delta}[T_i]_{\gamma\delta}+[C]_{\beta\delta}\delta_{\alpha\gamma}[T_i]_{\gamma\delta}=[BT_i]_{\alpha\beta}+[T_iC^T]_{\alpha\beta}.
\]
Thus, one has
\[
AT_i=BT_i+T_iC^T.
\]
Finally we substitute this $A$ into system \eqref{C-1} to obtain
\begin{align*} 
\begin{cases}
\frac{d}{dt}T_j =BT_j +T_j C^T+\kappa_{1}(T_cT_j^* T_j -T_j T_c^* T_j)+\kappa_{2}(T_j T_j^* T_c-T_j T_c^* T_j),\\
T_j(0)=T_j^0\quad\mbox{for all $j=1, 2, \cdots, N$}
\end{cases}
\end{align*}
and we have the following solution splitting property under the condition of property \ref{P3.2}:
\begin{align*} 
\begin{cases}
S_j=(e^{-Bt})T_j(e^{-C^Tt}), \quad j = 1, \cdots, N,  \\
{\dot S}_j =\kappa_{1}(S_cS_j^* S_j - S_j S_c^* S_j)+\kappa_{2}(S_j S_j^*S_c-S_j S_c^* S_j).
\end{cases}
\end{align*}

\section{The generalized Lohe matrix model}  \label{sec:4}
\setcounter{equation}{0}
In this section, we present emergent behaviors of the generalized Lohe matrix model with the same free flow dynamics. In this case, by solution splitting property, we may assume $A = 0$ in \eqref{C-1}:
\begin{equation} \label{D-1}
\begin{cases} 
\displaystyle {\dot T}_j =\kappa_{1}(T_cT_j^* T_j-T_jT_c^* T_j)+\kappa_{2}(T_jT_j^* T_c-T_jT_c^* T_j), \quad t > 0, \\
\displaystyle T_j(0) = T_j^0, \quad j = 1, \cdots, N.
\end{cases}
\end{equation}
Below, we estimate the time-evolution of $T_j^* T_j$ and $T_j T_j^*$ as follows.
\begin{lemma} \label{L4.1}
Let $\{ T_j \} \subset {\mathcal T}_2(\bbc; d_1 \times d_2)$ be a solution to \eqref{D-1}. Then, $T_j^* T_j$ and $T_j T_j^*$ satisfy
\begin{eqnarray*}
\frac{d}{dt}(T_j^* T_j) &=& \kappa_{2}(T_c^* T_j-T_j^* T_c)T_j^* T_j+\kappa_{2}T_j^* T_j(T_j^* T_c-T_c^* T_j), \\
\frac{d}{dt}(T_jT_j^*) &=& \kappa_{1}(T_cT_j^*-T_jT_c^*)T_jT_j^*+\kappa_{1}T_jT_j^*(T_jT_c^*-T_cT_j^*).
\end{eqnarray*}
\end{lemma}
\begin{proof}
By direct calculation, one has
\begin{align*}
\frac{d}{dt}(T_j^* T_j)&= {\dot T}_j^*T_j+T_j^* {\dot T}_j \\
&=\kappa_{1}(T_j^* T_j T_c^*-T_j^* T_cT_j^*)T_j +\kappa_{2}(T_c^* T_jT_j^*-T_j^* T_c T_j^*)T_j \\
&+\kappa_{1}T_j^*(T_cT_j^* T_j -T_jT_c^* T_j)+\kappa_{2}T_j^*(T_jT_j^* T_c-T_jT_c^* T_j)\\
&=\kappa_{2}(T_c^* T_jT_j^* T_j-T_j^* T_cT_j^* T_j +T_j^* T_j T_j^* T_c-T_j^* T_jT_c^* T_j)\\
&=\kappa_{2}(T_c^* T_j-T_j^* T_c)T_j^* T_j +\kappa_{2}T_j^* T_j(T_j^* T_c-T_c^* T_j).
\end{align*}
Similarly, one has
\begin{align*}
\frac{d}{dt}(T_jT_j^*)&= {\dot T}_j T_j^*+T_j {\dot T}_j^* \\
&=\kappa_{1}(T_cT_j^* T_j-T_jT_c^* T_j)T_j^*+\kappa_{2}(T_jT_j^* T_c-T_jT_c^* T_j)T_j^*\\
&+\kappa_{1}T_j(T_j^* T_jT_c^*-T_j^* T_cT_j^*)+\kappa_{2}T_j (T_c^* T_jT_j^*-T_j^* T_c T_j^*)\\
&=\kappa_{1}(T_cT_j^* T_jT_j^*-T_jT_c^* T_jT_j^*+T_jT_j^* T_jT_c^*-T_jT_j^* T_cT_j^*)\\
&=\kappa_{1}(T_cT_j^*-T_jT_c^*)T_jT_j^*+\kappa_{1}T_jT_j^*(T_jT_c^*-T_cT_j^*).
\end{align*}
\end{proof}
As a direct application of Lemma \ref{L4.1}, one has the following corollary.
\begin{corollary}
Let $\{ T_j \} \subset {\mathcal T}_2(\bbc; d_1 \times d_2)$ be a solution of \eqref{D-1}. Then, one has the conservation laws:
\begin{enumerate}
\item If $T_j(0) T_j(0)^*=I_{d_1} $, then we have $T_j(t) T_j(t)^*=I_{d_1}$ is constant of motion.
\item If $T_j(0)^* T_j(0)=I_{d_2} $, then we have $T_j(t)^* T_j(t)=I_{d_2}$ is constant of motion.
\end{enumerate}
\end{corollary}
\noindent For a configuration $\{ T_j \}$ with $\|T_j \|_F = 1$, we define
\[
T_c :=\frac{1}{N}\sum_{k=1}^NT_k, \quad {\mathcal V}[T] :=\frac{1}{N}\sum_{k=1}^N||T_k-T_c||_F^2, \quad \rho := \|T_c \|_F.
\]
For the Lohe matrix and Lohe sphere models, the functional ${\mathcal V}[T]$ corresponds to the potential in a gradient flow formulation (see \cite{H-K-R} for details).
\begin{lemma} \label{L4.2}
Let $\{ T_j \} \subset {\mathcal T}_2(\bbc; d_1 \times d_2)$ be a solution to \eqref{D-1} with $\|T_j^0 \|_F = 1$. Then, one has 
\begin{eqnarray*}
&& (i)~ {\mathcal V}[T] = 1 - \rho^2, \cr
&& (ii)~\frac{d}{dt} {\mathcal V}[T(t)]=-\frac{\kappa_{1}}{N}\sum_{j=1}^N||T_j T_c^*-T_cT_j^*||_F^2-\frac{\kappa_{2}}{N}\sum_{j=1}^N||T_j^* T_c-T_c^* T_j ||_F^2,  \cr
&& (iii)~\frac{d^2}{dt^2} {\mathcal V}[T] =-\frac{\kappa_{1}}{N}\sum_{j=1}^N\frac{d}{dt} \Big \langle{T_j T_c^*-T_cT_j^*, T_jT_c^*-T_cT_j^*} \Big \rangle_F 
\\
&& \hspace{2.5cm} -\frac{\kappa_{2}}{N}\sum_{j=1}^N\frac{d}{dt} \Big \langle{T_j^* T_c-T_c^* T_j, T_j^* T_c-T_c^* T_j} \Big \rangle_F.
\end{eqnarray*}
\end{lemma}
\begin{proof} 
(i)~It follows from Lemma \ref{L2.1} that 
\begin{equation} \label{D-2}
 \| T_j(t) \|_F = \|T_j^0 \|_F = 1, \quad t \geq 0,~~j = 1, \cdots, N.
\end{equation} 
Then, we use \eqref{D-2} to see
\begin{align}
\begin{aligned} \label{D-3}
{\mathcal V}(T)&=\frac{1}{N}\sum_{k=1}^N||T_k-T_c||^2_F\\
&=\frac{1}{N}\sum_{k=1}^N\mathrm{tr}((T^*_k-T^*_c)(T_k-T_c))=\frac{1}{N}\sum_{k=1}^N\mathrm{tr}(T_k^* T_k - T_c^* T_k - T_k^* T_c + T_c^* T_c)\\
&=\frac{1}{N}\sum_{k=1}^N\mathrm{tr}(T^*_kT_k-T^*_cT_c)=\frac{1}{N}\sum_{k=1}^N||T_k||_F^2-||T_c||_F^2 = 1  - \rho^2.
\end{aligned}
\end{align}
(ii)~We differentiate relation \eqref{D-3} to get 
\begin{align}
\begin{aligned} \label{D-4}
\frac{d}{dt} {\mathcal V}(T(t)) =\frac{d}{dt}\left(\frac{1}{N}\sum_{k=1}^N||T_k^0||_F^2-||T_c(t)||_F^2\right) =-\frac{d}{dt}||T_c(t)||_F^2=-\langle{\dot{T}_c, T_c}\rangle_F-\langle{T_c, \dot{T}_c}\rangle_F,
\end{aligned}
\end{align}
where $\langle{A, B}\rangle_F =\sum_{\alpha\beta}[\bar{A}]_{\alpha\beta}[B]_{\alpha\beta}=\mathrm{tr}(A^* B)$. \newline

On the other hand, $T_c$ satisfies
\begin{equation} \label{D-5}
\dot{T}_c=\frac{1}{N}\sum_{j=1}^N\Big[ \kappa_{1}(T_cT_j^* T_j-T_jT_c^* T_j)+\kappa_{2}(T_jT_j^* T_c-T_jT_c^* T_j)\Big ].
\end{equation}
Now, we combine \eqref{D-4} and \eqref{D-5} to obtain 
\begin{align}
\begin{aligned} \label{D-6}
&\frac{d}{dt} {\mathcal V}(T) =-\frac{1}{N}\sum_{j=1}^N \Big (\kappa_{1}\langle{T_cT_j^* T_j-T_jT_c^* T_j, T_c}\rangle+\kappa_{2}\langle{T_jT_j^* T_c-T_jT_c^* T_j, T_c}\rangle \Big)+\mbox{(c. c.)}\\
&=-\frac{1}{N}\sum_{j=1}^N\Big(\kappa_{1}\mathrm{tr}(T_j^* T_j T_c^* T_c-T_j^* T_cT_j^* T_c)+
\kappa_{2}\mathrm{tr}(T_c^* T_j T_j^* T_c-T_j^* T_c T_j^* T_c)\Big)+\mbox{(c. c.)}\\
&=-\frac{\kappa_{1}}{N}\sum_{j=1}^N\langle{T_jT_c^*-T_cT_j^*, T_jT_c^*-T_cT_j^*}\rangle_F-\frac{\kappa_{2}}{N}\sum_{j=1}^N\langle{T_j^* T_c-T_c^* T_j, T_j^* T_c-T_c^* T_j}\rangle_F\\
&=-\frac{\kappa_{1}}{N}\sum_{j=1}^N ||T_jT_c^*-T_cT_j^*||_F^2-\frac{\kappa_{2}}{N}\sum_{i=1}^N||T_j^* T_c-T_c^* T_j||_F^2.
\end{aligned}
\end{align}
\noindent (iii)~We differentiate \eqref{D-6} again with respect to $t$ to get 
\begin{align}
\begin{aligned} \label{D-7}
\frac{d^2}{dt^2}{\mathcal V}[T] &=-\frac{\kappa_{1}}{N}\sum_{j=1}^N\frac{d}{dt} \Big \langle{T_j T_c^*-T_cT_j^*, T_jT_c^*-T_cT_j^*} \Big \rangle_F \\
&\hspace{0.5cm} -\frac{\kappa_{2}}{N}\sum_{j=1}^N\frac{d}{dt} \Big \langle{T_j^* T_c-T_c^* T_j, T_j^* T_c-T_c^* T_j} \Big \rangle_F.
\end{aligned}
\end{align}
Note that  the right hand side of \eqref{D-7} can be expressed in terms of $\dot{T}_j, \dot{T}_c, T_j$ and $T_c$. Thus $\frac{d^2}{dt^2} {\mathcal V}(T)$ is uniformly bounded.
\end{proof}
\begin{theorem}\label{T4.1}
Let $\{ T_j \} \subset {\mathcal T}_2(\bbc; d_1 \times d_2)$ be a solution to \eqref{D-1} with $\|T_j^0 \|_F = 1$. Then, the following estimates hold. 
\begin{enumerate}
\item
There exists a positive value ${\mathcal V}_\infty$ such that 
\[ \lim_{t \to \infty+} {\mathcal V}[T] = {\mathcal V}_\infty.  \]
\item
The orbital derivative of ${\mathcal V}(T)$ tends to zero: 
\[ \lim_{t\rightarrow\infty} \frac{d}{dt} {\mathcal V}[T(t)] = 0, \quad \lim_{t\rightarrow\infty} \Big( ||T_j T_c^*-T_cT_j^*||_F + ||T_j^*T_c-T_c^*T_j ||_F \Big) = 0. \]
\item
The following estimate hold.
\[  \frac{1}{N} \sum_{j=1}^N \Big[ \kappa_{1} \int_0^\infty ||T_j T_c^*-T_cT_j^*||_F^2 dt+ \kappa_{2}\int_0^\infty ||T_j^*T_c-T_c^*T_j ||_F^2dt \Big] \leq 1-||T_c^0||_F^2. \] 
\end{enumerate}
\end{theorem}
\begin{proof}
(i)~It follows from (ii) in Lemma \ref{L4.2} that ${\mathcal V}$ is non-increasing and it is bounded below by $0$. Thus, it converges to a nonnegative value.  \newline

\noindent (ii) First, we use the result in (i) and Barbalat's lemma to obtain
\[
\lim_{t\rightarrow\infty}\frac{d}{dt} {\mathcal V}[T(t)]=0.
\]
On the other hand, note that 
\begin{equation}\label{D-8}
\frac{d}{dt}{\mathcal V}[T(t)]=-\frac{\kappa_{1}}{N}\sum_{j=1}^N||T_jT_c^*-T_cT_j^*||_F^2-\frac{\kappa_{2}}{N}\sum_{j=1}^N ||T_j^* T_c-T_c^* T_j ||_F^2.
\end{equation}
This yields
\[
||T_j T_c^*-T_cT_j^*||_F\rightarrow0,\quad ||T_j^*T_c-T_c^*T_j ||_F\rightarrow0 \quad \mbox{as $t \to \infty$}.
\]
 \newline

\noindent (iii)~We integrate \eqref{D-8} from $t=0$ to $t=\tau$ to obtain 
\begin{align}
\begin{aligned} \label{D-9}
& {\mathcal V}[T(\tau)] -{\mathcal V}[T(0)] \\
& \hspace{0.5cm} =-\frac{\kappa_{1}}{N}\int_0^\tau\sum_{i=1}^N||T_j T_c^*-T_cT_j^*||_F^2dt-\frac{\kappa_{2}}{N}\int_0^\tau\sum_{i=1}^N||T_j^*T_c-T_c^*T_j ||_F^2dt.
\end{aligned}
\end{align}
On the other hand, we use (i) of Lemma \ref{L4.1} to get
\begin{align}
\begin{aligned} \label{D-10}
&{\mathcal V}[T(\tau)]-{\mathcal V}[T(0)] \\
& \hspace{0.5cm} =\left(\frac{1}{N}\sum_{k=1}^N||T_k(\tau)||^2_F-||T_c(\tau)||_F^2\right)-\left(\frac{1}{N}\sum_{k=1}^N||T_k(0)||^2_F-||T_c(0)||_F^2\right)\\
& \hspace{0.5cm} =||T_c(0)||^2_F-||T_c(\tau)||_F^2.
\end{aligned}
\end{align}
Then, it follows from \eqref{D-9} and \eqref{D-10} that
\[
\frac{\kappa_{1}}{N}\int_0^\tau\sum_{j=1}^N||T_j T_c^*-T_cT_j^*||_F^2dt+\frac{\kappa_{2}}{N}\int_0^\tau\sum_{j=1}^N ||T_j^*T_c-T_c^*T_j ||_F^2dt=||T_c(\tau)||_F^2-||T_c^0||_F^2.
\]
Now, we let $\tau\rightarrow\infty$ and $||T_c(\tau)||_F \leq 1$ to see the desired estimate:
\[
\frac{\kappa_{1}}{N}\int_0^\infty\sum_{j=1}^N||T_jT_c^*-T_cT_j^*||_F^2dt+\frac{\kappa_{2}}{N}\int_0^\infty\sum_{j=1}^N||T_j^*T_c-T_c^*T_j||_F^2dt\leq 1-||T_c^0||_F^2.
\] 
\end{proof}
\begin{remark}\label{R4.1}
It follows from Theorem \ref{T4.1} that for each $i=1, 2, \cdots, N$, we can see
\[
||T_iT_c^*-T_cT_i ||_F \in L^2((0, \infty)),\quad ||T_i^*T_c-T_c^*T_i ||_F\in L^2((0, \infty)).
\]
\end{remark}

Since two coupling terms involve with $\kappa_1$ and $\kappa_2$ are symmetric, we will only consider ‘‘{\it the reduced Lohe matrix model} " with $\kappa_2=0$:
\[ {\dot T}_j=A_jT_j +\kappa_{1}(T_cT_j^*T_j-T_jT_c^*T_j). \]

\section{Weak emergent estimate for the reduced Lohe matrix model} \label{sec:5}
\setcounter{equation}{0} 
In this section, we study a reformulation of the reduced Lohe matrix model on the restricted class of initial data into the Lohe matrix model with diagonal frustration and provide a weak emergent estimate. \newline

Consider the reduced Lohe matrix model:
\begin{align}\label{E-1}
\begin{cases}
{\dot T}_i=\kappa_{1}(T_cT_i^* T_i -T_i T_c^* T_i), \quad t > 0, \\
T_i(0)=T_i^0 \in {\mathcal T}_2(\bbc; d_1 \times d_2), \quad i = 1, \cdots, N,
\end{cases}
\end{align}
subject to the initial conditions:
\begin{equation} \label{E-1-0}
T_i^{0*} T_i^0 = T_j^{0*} T_j^0, \quad 1 \leq i, j \leq N. 
\end{equation}
In the sequel, we will provide a procedure to rewrite \eqref{E-1} as the Lohe matrix model with a diagonal frustration for a suitable unitary matrix $U_i$:
\begin{equation}
\begin{cases} \label{E-1-1}
\displaystyle \dot{U}_i=\frac{\kappa_{1}}{N}\sum_{k=1}^N((U_kD)-U_i(U_kD)^*U_i), \quad t > 0, \\
\displaystyle U_i(0) = U_i^0, \quad i = 1, \cdots, N,
\end{cases}
\end{equation}
where $D$ is a diagonal matrix, and discuss the equivalence between \eqref{E-1} and \eqref{E-1-1} in terms of the complete aggregation. The reason for this reformulation is to use existing tool box from earlier results \cite{H-K, H-K-P-Z, H-K-R} for the Lohe matrix model.

\subsection{A reformulation into the Lohe matrix model} \label{sec:5.1} In this subsection, we present a reformulation of \eqref{E-1} as the Lohe matrix model with a frustration \eqref{E-1-1}. For this, we use a singular value decomposition (SVD) of $T_i$ in the sequel. 
\begin{proposition} \label{P5.1}
\emph{(Conservation of ranks)}
Let $\{T_i \} \subset {\mathcal T}_2(\bbc; d_1 \times d_2)$ be a solution to \eqref{E-1} - \eqref{E-1-0}. Then, the rank of the $T_i\in \mathbb{C}^{d_1\times d_2}$ is a conserved quantity for all $i=1, 2, \cdots, N$:
\[ \mbox{Rank}~(T_i(t)) = \mbox{Rank}~(T_i^0), \quad t > 0,~~i = 1, \cdots, N. \]
\end{proposition}
\begin{proof}
We set 
\[ \mathrm{Rank}~(T_i^0)=r \leq \min\{d_1, d_2 \}. \]
Then, we set $\{u_1(0), u_2(0), \cdots, u_r(0)\}\subset \mathbb{C}^{d_2}$ be the orthogonal basis of the row space $\mbox{row}(T^0_i)$. Then there are $\{v_1, v_2, \cdots, v_r\}\in\mathbb{C}^{d_1}$ such that
\[
T^0_i v_j=u_j(0)
\]
for all $j=1, 2, \cdots, r$. Since $\{v_1, v_2, \cdots, v_r\}$ are linearly independent, we can extend this subset to a basis of $\mathbb{C}^{d_1}$. We denote this basis as 
\[
\{v_1, v_2, \cdots, v_r, v_{r+1}, \cdots, v_{d_1}\}.
\]
By linear algebra, we can see that 
\[
\{v_{r+1}, \cdots, v_{d_1}\}
\]
is a basis of $\mathrm{null}(T^0_i)$. Now, we claim:
\begin{equation} \label{E-2}
\langle{T_i(t) x, T_i(t) y}\rangle~\mbox{is a conserved quantity} \quad \mbox{for  all}~x, y\in\mathbb{C}^{d_1}.
\end{equation}
{\it Proof of \eqref{E-2}}:~Let  $x, y\in\mathbb{C}^{d_1}$. Then, since $T_i^* T_i$ is a conserved quantity, one has
\[
\langle{T_i(t) x, T_i(t) y}\rangle=x^* T_i(t)^* T_i(t) y=x^* T_i(0)^* T_i(0) y=\langle{T_i(0)x, T_i(0)y}\rangle,
\]
which verifies \eqref{E-2}. Next, we set 
\begin{equation} \label{E-3}
u_j(t)=T_i(t)v_j
\end{equation}
for all $j=1, 2, \cdots, r$. Then we use \eqref{E-2} and \eqref{E-3} to get 
\[
\langle{u_j(t), u_k(t)}\rangle=\langle{u_j(0), u_k(0)}\rangle=\delta_{jk}.
\]
For all $j=r+1, \cdots, d_1$, one has 
\[
T_i(t)v_j=0,
\]
This and \eqref{E-2} yield
\[
\langle{T_i(t)v_j, T_i(t)v_j}\rangle=\langle{T_i(0)v_j, T_i(0)v_j}\rangle=\langle{0, 0}\rangle=0.
\]
Hence, we can conclude that the rank of the $T_i(t)$ is also $r$. This means that the rank of the $T_i(t)$ is a conserved quantity,
\end{proof}

\vspace{0.5cm}

By the singular value decomposition of $T_i\in\mathbb{C}^{d_1\times d_2}$, one has
\[
T_i=U_i \Sigma_i V_i^*
\]
where $U_i$ and $V_i$ are $d_1 \times d_1$ unitary matrix and $d_2 \times d_2$ unitary matrices, respectively, and $\Sigma_i$ is a $d_1\times d_2$ rectangular diagonal matrix with non-negative real numbers on the diagonal:
\[
[\Sigma_i]_{jk}=
\begin{cases}
\lambda_j\qquad\mbox{when $j=k\leq\min \{d_1, d_2 \}$,}\\
0\qquad \mbox{otherwise.}
\end{cases}
\]

\begin{theorem} \label{T5.1}
Let $\{T_j\}  \subset {\mathcal T}_2(\bbc; d_1 \times d_2)$ be a solution of system \eqref{E-1}. Then, the singular value decomposition of $T_i(t)$ is given as follows.
\[
T_i(t)=U_i(t)\Sigma_i V_i^*,
\] 
where $U_i$ is a time-dependent, and $\Sigma_i,~V_i$ are time-independent.
\end{theorem}
\begin{proof} We split its proof into three steps. \newline

\noindent $\bullet$~Step A (Selection of $\Sigma_i$ and $V_i$):  Since $T_i(0)^* T_i(0)$ is  hermitian, we can find a diagonalization of $T_i(0)^* T_i(0)$ as follows:
\[
T_i(0)^* T_i(0)=V_i^* D_i V_i.
\]
We use hermitian property of $T_i(0)^* T_i(0)$ to see that the diagonal matrix $D_i$ is a real valued matrix:
\[
D_i=\mathrm{diag}(\mu_1,\mu_2, \cdots, \mu_r, 0, 0, \cdots, 0)
\]
where $r=\mathrm{Rank}~(T_i)$ and $\mu_1, \cdots, \mu_r$ are nonzero number. We claim:
\begin{equation} \label{E-3-1}
 \mu_i > 0, \quad i = 1, \cdots, r. 
\end{equation} 
Let $v_j$ be a unit eigenvector corresponds to $\mu_j$. Then we have
\[
\mu_j=\langle{v_j, D_i v_j}\rangle=\langle{v_j, T_i^*(0) T_i(0) v_j}\rangle=\langle{T_i(0) v_j, T_i(0) v_j}\rangle>0,
\] 
which implies \eqref{E-3-1}. Then we can express $D_i=\Sigma^*_i\Sigma_i$ with
\[
[\Sigma_i]_{jk}=
\begin{cases}
\lambda_j:=\sqrt{\mu_j}\qquad\mbox{when $j=k\leq\max(d_1, d_2)$,}\\
0\qquad \mbox{otherwise.}
\end{cases}
\]

\vspace{0.2cm}

\noindent $\bullet$~Step B (Selection of $U_i$): We choose a unitary matrix $U_i^0$ to satisfy
\[
T_i(0)=U_i^0\Sigma_i V_i^*.
\]
Then, for such $\{U_i^0 \}$, we define a matrix $U_i = U_i(t)$ as a solution of following ODE system:
\begin{align} \label{E-4}
\begin{cases}
\dot{U}_i(t)U_i(t)^*=\kappa_{1}(T_c(t)T_i(t)^*-T_i(t)T_c(t)^*), \quad t>0,\\
U_i(0)=U_i^0.
\end{cases}
\end{align}
The unitarity of $U_i$ can be seen easily as follows. First note that
\[
\frac{d}{dt}(U_iU_i^*)=\kappa_{1}(T_cT_i^*-T_iT_c^*)+\kappa_{1}(T_iT_c^*-T_cT_i^*)=0.
\]
Thus,
\[ U_i(t) U_i^*(t) = U_i^0 (U^0_i)^* = I_{d_1}, \quad t \geq 0. \] 

\vspace{0.2cm}

\noindent $\bullet$~Step C (Choice of $U_i$): It follows from Step A and Step B that 
\begin{equation} \label{E-4-1}
T_i(t)=U_i(t)\Sigma_i V_i^*.
\end{equation}
Finally, we check the ansatz \eqref{E-4-1} satisfies \eqref{E-1}. By direct calculation, one has 
\[ {\dot T}_i =\dot{U}_i(t)\Sigma_i V_i^*=\kappa_{1}(T_cT_i^*-T_iT_c^*)U_i\Sigma_iV_i^* =\kappa_{1}(T_cT_i^*T_i-T_iT_c^*T_i). \]
\end{proof}
\begin{remark}
Note that it follows from the condition \eqref{E-1-0}:
\[
T_i^{0*}T_i^0 =T_j^{0*} T_j^0, \quad \forall~i, j
\]
and the procedures of  SVD of $T_i^0$ and $T_j^0$, we can choose $V_i$ and $\Sigma_i$ such that 
\begin{equation*} \label{E-4-2}
 V_i = V, \quad \Sigma_i = \Sigma, \quad \forall~i.
 \end{equation*}
\end{remark}

Then, we substitute the ansatz
\[ T_i(t)=U_i(t)\Sigma_i V_i^*, \quad i = 1, \cdots, N, \]
into \eqref{E-4} to get
\begin{align*}
\dot{U}_iU_i^*&=\frac{\kappa_{1}}{N}\sum_{k=1}^N(T_kT_i^*-T_iT_k^*) =\frac{\kappa_{1}}{N}\sum_{k=1}^N(U_k\Sigma_kV_k^*V_i\Sigma_i^*U_i^*-U_i\Sigma_iV_i^*V_k\Sigma^*_kU_k^*)\\
&=\frac{\kappa_{1}}{N}\sum_{k=1}^N(U_k\Sigma_k\Sigma_i^*U_i^*-U_i\Sigma_i\Sigma_k^*U_k^*).
\end{align*}
Thus, one has the Lohe matrix model with diagonal frustration:
\begin{align} \label{E-5}
\dot{U}_i=\frac{\kappa_{1}}{N}\sum_{k=1}^N((U_kD)-U_i(U_kD)^*U_i),
\end{align}
where $D=\mathrm{diag}(\lambda_1^2, \lambda_2^2, \cdots, \lambda_r^2, 0, 0, \cdots, 0)$ is diagonal matrix with $\mbox{rank} (T_i) = r$, and 
\[
[\Sigma]_{jk}=
\begin{cases}
\lambda_j\qquad\mbox{when $j=k\leq r$,}\\
0\qquad \mbox{otherwise.}
\end{cases}
\]
In the sequel, we will show that systems \eqref{E-1} and \eqref{E-5} are equivalent asymptotically from the view point of complete aggregation.

\noindent Next, we study an elementary lemma as follows.
\begin{lemma} \label{L5.1}
Let $A$ and $B$ be matrices with the same size $d_1\times d_2$. Then we have
\[
||A^*A-B^*B||_F\leq ||A-B||_F\cdot||A+B||_F.
\] 
\end{lemma}
\begin{proof}
Note that 
\[
A^*A-B^*B=\frac{1}{2}\Big[ (A+B)^*(A-B)+(A-B)^*(A+B) \Big ].
\]
This yields
\begin{equation} \label{E-6}
||A^*A-B^*B||_F\leq\frac{1}{2}(||(A+B)^*(A-B)||_F+||(A-B)^*(A+B)||_F).
\end{equation}
We use the Cauchy-Schwarz inequality to find
\begin{align*}
||X^*Y||_F^2&=\sum_{i, k}\left|\sum_j \bar{[X]}_{ij}[Y]_{jk}\right|^2\leq\sum_{i, k}\left(\sum_{j}|[X]_{ij}|^2\cdot\sum_{j}|[Y]_{jk}|^2\right)\\
&=\sum_{i, j}|[X]_{ij}|^2\cdot\sum_{j, k}|[Y]_{jk}|^2=||X||_F^2\cdot||Y||_F^2,
\end{align*}
i.e., 
\begin{equation} \label{E-7}
||X^*Y||_F\leq||X||_F\cdot||Y||_F.
\end{equation}
Finally, we combine \eqref{E-6} and \eqref{E-7} to get
\[ ||A^*A-B^*B||_F\leq||A+B||_F\cdot||A-B||_F. \]
\end{proof}
\begin{theorem} \label{T5.2}
Let $\{T_j\} \subset {\mathcal T}_2(\bbc; d_1 \times d_2)$ be a solution of system \eqref{E-1}.  If the complete aggregation emerges asymptotically
\begin{equation} \label{E-7-1}
\lim_{t \to \infty} \max_{1 \leq i,j \leq N} ||T_i(t)-T_j(t)||_F = 0,
\end{equation}
then, we can choose $V$ and $\Sigma$ satisfy the following relations:
\begin{equation*} \label{E-7-2}
T_i=U_i(t)\Sigma_i V_i^*,\quad V_i=V \quad \mbox{and} \quad \Sigma_i=\Sigma \quad \mbox{for all ~$i=1, 2, \cdots, N$}. 
\end{equation*} 
\end{theorem}
\begin{proof} Suppose that the complete aggregation \eqref{E-7-1} holds. Then, it follows from Lemma \ref{L5.1} that 
\[
||T_i^*T_i-T_j^*T_j||_F\leq||T_i+T_j||_F\cdot||T_i-T_j||_F\leq (||T_i||_F+||T_j||_F)\cdot||T_i-T_j||_F.
\]
We use the relation \eqref{E-7-1} and the conservation of $||T_i||_F$ and $||T_j||_F$ to get 
\begin{equation} \label{E-7-3}
||T_i^*T_i-T_j^*T_j||_F\rightarrow0
\end{equation}
as time goes infinity. On the other hand, we use conservation laws for $T_i^*T_i$ and $T_j^*T_j$ and \eqref{E-7-3} to find
\[
 ||T_i(0)^*T_i(0)-T_j(0)^*T_j(0)||_F = ||T_i(t)^*T_i(t) -T_j(t)^*T_j(t)||_F \to 0, 
\]
as $t \to \infty$. This yields
\begin{equation*} \label{E-7-4}
 ||T_i(0)^*T_i(0)-T_j(0)^*T_j(0)||_F  = 0, \quad \mbox{i.e.,} \quad T_i(0)^*T_i(0)=T_j(0)^*T_j(0).
\end{equation*}
From the SVD, we can decompose $T_i(0)$ and $T_j(0)$ as follows:
\begin{equation*} \label{E-7-5}
T_i(0)=U_i(0) \Sigma V^*,\qquad T_j(0)=U_j(0)\Sigma V^*,
\end{equation*}
for some common matrices $\Sigma$ and $V$.
\end{proof}

\subsection{Equivalence relation} 
Consider the reduced Lohe matrix model \eqref{E-1} - \eqref{E-1-0}:
\begin{align}\label{E-8-1}
\begin{cases}
{\dot T}_i=\kappa_{1}(T_cT_i^* T_i-T_iT_c^* T_i), \quad t > 0,\\
T_i(0)=U_i^0\Sigma V^*, \quad i = 1, \cdots, N,
\end{cases}
\end{align}
with $d_1\times d_1$ size unitary matrix $U_i$, $d_2\times d_2$ size unitary matrix $V$ and $d_1\times d_2$ size matrix $\Sigma$ with
\[
[\Sigma]_{jk}=
\begin{cases}
\lambda_j\geq0\qquad\mbox{when $j=k\leq r$,}\\
0\qquad \mbox{otherwise}
\end{cases},
\]
and the Lohe matrix model with a diagonal frustration: 
\begin{align}\label{E-8-2}
\begin{cases}
\dot{U}_i=\kappa_{1}(U_cD-U_iD^*U_c^*U_i), \quad t > 0, \\
U_i(0)=U_i^0, \quad i = 1, \cdots, N,
\end{cases}
\end{align}
where
\[
D=\Sigma\Sigma^*, \quad D=\mathrm{diag}(\lambda^2_1, \lambda^2_2, \cdots, \lambda^2_{r}, \lambda_{r+1}^2=0, \cdots, \lambda_{d_1}^2=0).
\]
\noindent Next we study the equivalence between \eqref{E-8-1} and \eqref{E-8-2} in terms of complete aggregation.  Recall the necessary condition for the complete aggregation (see Theorem \ref{T5.2}):
\[ V_i=V \quad \mbox{and} \quad \Sigma_i=\Sigma \quad \mbox{for all $i=1, 2, \cdots, N$.}  \]
Thus, reformulated system reads as 
\begin{equation} \label{E-9}
\dot{U}_i=\kappa_{1}(U_c\Sigma\Sigma^*-U_i\Sigma\Sigma^*U_c^*U_i).
\end{equation}
In the sequel, we study the complete synchronization in terms of $U_i$ and $U_j$. Since $D=\Sigma\Sigma^*$ is a diagonal matrix, we can set
\[
D=\mathrm{diag}(\lambda^2_1, \lambda^2_2, \cdots, \lambda^2_{r}, \lambda_{r+1}^2=0, \cdots, \lambda_{d_1}^2=0) \in \bbc^{d_1 \times d_1}.
\]
Next, we study the conditions for emergence of \eqref{E-9}.  Here $\bbu(d_1)$ is the unitary group consisting of $d_1 \times d_1$ unitary matrices.
\begin{lemma}\label{L5.2}
Let $\{U_i \} \subset \bbu(d_1)$ be a solution to \eqref{E-9} and $T_i(t)$ be given by \eqref{E-4-1}. Then, we have the following assertions:
\begin{enumerate}
\item
The following relation holds:
 \[ ||T_i-T_j||_F  = ||(U_i-U_j)\Sigma||_F. \]
 \item
For $i \not = j$, one has 
\[
\Big( \min_{1\leq k \leq r} \lambda^2_k \Big) \cdot ||U_i-U_j||^2\leq||(U_i-U_j)\Sigma||_F^2\leq \Big( \max_{1\leq k \leq r}\lambda^2_k \Big) \cdot ||U_i-U_j||^2.
\]
\end{enumerate}
\end{lemma}
\begin{proof} 
(i)~By \eqref{E-9}, we can express $T_i-T_j$ as 
\[
T_i-T_j=(U_i-U_j)\Sigma V^*.
\]
Then, we can calculate the Frobenius norm of this matrix.
\begin{align*}
||T_i-T_j||_F^2&=\mathrm{tr}((T_i-T_j)^*(T_i-T_j))=\mathrm{tr}(V\Sigma^*(U_i-U_j)^*(U_i-U_j)\Sigma V^*)\\
&=\mathrm{tr}(\Sigma^*(U_i-U_j)^*(U_i-U_j)\Sigma)=||(U_i-U_j)\Sigma||_F^2.
\end{align*}
(ii)~By direct calculation, one has
\begin{align*}
\begin{aligned}
||(U_i-U_j)\Sigma||_F^2&=\mathrm{tr}(\Sigma^*(U_i-U_j)^*(U_i-U_j)\Sigma)\\
&=[\Sigma^*]_{\alpha\beta}[(U_i-U_j)^*]_{\beta\gamma}[(U_i-U_j)]_{\gamma\delta}[\Sigma]_{\delta\alpha} \\
&=[\Sigma\Sigma^*]_{\delta\beta}[(U_i-U_j)^*]_{\beta\gamma}[U_i-U_j]_{\gamma\delta}.
\end{aligned}
\end{align*}
This implies
\[
||(U_i-U_j)\Sigma||_F^2=\sum_{k=1}^{d_1}\lambda^2_k [(U_i-U_j)^*]_{k\gamma}[U_i-U_j]_{\gamma k}=\sum_{j, k = 1}^{r}\lambda_k^2|[U_i-U_j]_{jk}|^2.
\]
Finally, we can obtain the desired estimates:
\[
\Big( \min_{1 \leq k \leq N} \lambda^2_k \Big) \cdot ||U_i-U_j||_F^2\leq||(U_i-U_j)\Sigma||_F^2=||T_i-T_j||_F^2\leq \Big( \max_{1 \leq k \leq N}\lambda^2_k \Big) \cdot ||U_i-U_j||_F^2.
\]
\end{proof}
Now, we are ready to discuss the equivalence between \eqref{E-8-1} and \eqref{E-8-2} as follows.
\begin{theorem}\label{T5.3}
The following assertions hold.
\begin{enumerate}
\item
If the complete aggregation occurs for \eqref{E-8-2}, then complete aggregation also occurs for system \eqref{E-8-1}. 
\item
Suppose that $\lambda_1, \lambda_2, \cdots, \lambda_{d_1}$ are nonzero constant. If the complete synchronization occurs for \eqref{E-8-1}, then complete synchronization also occurs for system \eqref{E-8-2}.
\end{enumerate}
\end{theorem}
\begin{proof}
\noindent (i)~Suppose that system \eqref{E-8-2} exhibits the complete aggregation asymptotically:
\begin{equation*} \label{E-10}
\lim_{t \to \infty} \max_{1 \leq i, j \leq N} \| U_i(t) - U_j(t) \|_F = 0.
\end{equation*}
Then, we use the relation in Lemma \ref{L5.2}:
\[
||T_i-T_j||_F\leq \Big( \max_{1 \leq k \leq r} |\lambda_k| \Big) \cdot ||U_i-U_j||_F
\]
to see the complete aggregation of \eqref{E-8-1}:
\[
||T_i(t)-T_j(t)||_F=0, \quad t \geq 0.
\]

\noindent (ii)~Suppose that $\lambda_1, \lambda_2, \cdots, \lambda_{d_1}$ are nonzero. Then we have 
\begin{equation*} \label{E-11}
 \min_{1 \leq k \leq N}  |\lambda_{k}| >0.
\end{equation*} 
It follows from Lemma \ref{L5.2} (ii) that 
\begin{equation*} \label{E-12}
\Big( \min_{1 \leq k \leq N} |\lambda_{k} | \Big) \cdot ||U_i-U_j||_F\leq||T_i-T_j||_F.
\end{equation*}
Thus, the complete aggregation for system \eqref{E-8-2} yields the complete aggregation of \eqref{E-8-1}.
\end{proof}
\begin{remark} \label{R5.3}
1. Note that if $\lambda_1, \lambda_2, \cdots, \lambda_{d_1}$ are nonzero, complete aggregations of system \eqref{E-8-1} and  \eqref{E-8-2} are equivalent. 

\noindent 2. The same statements can be made for the practical aggregation. 
\end{remark}

\subsection{Effect of diagonal frustration} Let $\lambda_1, \lambda_2, \cdots, \lambda_{d_1}$ be nonzero. Then, the emergent dynamics of \eqref{E-8-1} is equivalent to \eqref{E-8-2}.  \newline

Consider the Lohe matrix model with a diagonal frustration:
\begin{align}\label{E-12-1}
\begin{cases}
\dot{U}_i=\kappa_{1}(U_cD-U_iD^*U_c^*U_i), \quad t > 0, \\
U_i(0)=U_i^0, \quad i = 1, \cdots, N,
\end{cases}
\end{align}
where $U_i^0$ are unitary matrices and $D$ is a diagonal matrix with positive real diagonal entries.
\begin{lemma} \label{L5.3}
Suppose that $D=\mathrm{diag}(\lambda_1^2, \lambda_2^2, \cdots, \lambda_{d_1}^2)$ is a diagonal matrix  for positive real number $\lambda_1, \lambda_2, \cdots, \lambda_{d_1}$. Then, the square root of the matrix $D$ is given as follows.
\[
\sqrt{D}=\mathrm{diag}(\lambda_1, \lambda_2, \cdots, \lambda_{d_1}).
\]
\end{lemma}
\begin{remark}
Since $D$ and $\sqrt{D}$ are full-rank diagonal matrices, they are invertible. 
\end{remark}

Let $\{U_i\}_{i=1}^N$ be a solution of system \eqref{E-12-1} with $\|T_i \|_F = 1$. Recall the variance of the Lohe matrix model \eqref{E-12-1}: 
\[
{\mathcal V}[T]=\frac{1}{N}\sum_{k=1}^N||T_k-T_c||_F^2=\frac{1}{N}\sum_{k=1}^N||T_k||_F^2-||T_c||_F^2= 1-||T_c||_F^2, \quad \mbox{and} \quad  T_i = U_i \Sigma V^*. \]
Then, we set
\[ \tilde{{\mathcal V}}[U] := 1-||T_c||_F^2.   \]
By direct estimate, one has
\begin{align*}
\begin{aligned} \label{E-12-2}
\tilde{{\mathcal V}}[U] &=1-||U_c\Sigma V^*||_F^2 =1-\mathrm{tr}[U_c\Sigma V^* V\Sigma^* U_c^*] \\
&= 1-\mathrm{tr}[U_c D U_c^*]=M^2-\mathrm{tr}[(U_c\sqrt{D})(U_c\sqrt{D})^*] =1-||U_c\sqrt{D}||_F^2.
\end{aligned}
\end{align*}
Then it follows from Lemma \ref{L4.1} that we can obtain following lemma.
\begin{lemma} \label{L5.4}
Let $\{U_i\}_{i=1}^N$ be a solution of system \eqref{E-12-1} with $\|T_i \|_F = 1$. Then one has
\begin{equation} \label{E-12-3}
\frac{d}{dt}\tilde{{\mathcal V}}[U]  =-\frac{\kappa_{1}}{N}\sum_{i=1}^N|| (U_i\sqrt{D})(U_c\sqrt{D})^*-(U_c\sqrt{D})(U_i\sqrt{D})^*||_F^2.
\end{equation}
\end{lemma} 
\begin{proof}
It follows from Lemma \ref{L4.2} that 
\begin{equation} \label{E-13}
 \frac{d}{dt} {\mathcal V}[T(t)]=-\frac{\kappa_{1}}{N}\sum_{j=1}^N||T_j T_c^*-T_cT_j^*||_F^2. 
\end{equation} 
First, we express the term $T_iT_c^*-T_cT_i^*$ in terms of $U_i$, $U_c$ and $D$:
\begin{equation} \label{E-14}
T_iT_c^*-T_cT_i^*=U_i\Sigma V^* V \Sigma^* U_c^*-U_c\Sigma V^* V \Sigma^* U_i^*=U_iDU_c^*-U_cDU_i^*.
\end{equation}
Finally, we combine \eqref{E-13} and \eqref{E-14} to get the desired estimate.
\end{proof}
Next, we provide a weak emergent estimate by showing that the variance functional $\tilde{\mathcal V}[U]$ converges to some value asymptotically.  This weak estimate will be improved in next section. 
\begin{theorem} \label{T5.4}
Let $\{U_i\}_{i=1}^N$ be a solution of system \eqref{E-12-1} with $\|T_i \|_F = 1$. Then, we have 
\[ \exists~\lim_{t \to \infty} \tilde{\mathcal V}[U(t)]  \quad \mbox{and} \quad \lim_{t\rightarrow\infty}\frac{d}{dt} \tilde{\mathcal V}[U(t)]=0. \]
\end{theorem}
\begin{proof}
\noindent It follows from \eqref{E-12-3} in Lemma \ref{L5.4} that $\tilde{{\mathcal V}}[U(\cdot)]$ is decreasing and bounded below. Thus, $\tilde{{\mathcal V}}[U(\cdot)]$ converges as $t \to \infty$.  Now note that $U_i$ and ${\dot U}_i$ are uniformly bounded using \eqref{E-12-1}. We differentiate the $\frac{d}{dt}\tilde{{\mathcal V}}[U]$ with respect to $t$ to see that $|\frac{d^2}{dt^2}\tilde{{\mathcal V}}[U(t)] $ can be expressed in terms of $U_i$ and ${\dot U}_i$. Thus, we have the uniform boundedness of $\left|\frac{d^2}{dt^2}\tilde{{\mathcal V}}[U(t)]\right|$, i.e., there exits a positive constant $M$ such that 
\[
\left|\frac{d^2}{dt^2}\tilde{{\mathcal V}}[U(t)] \right|<M.
\]
Then, we can apply Barbalat's lemma using the above uniform boundedness to see the desired second estimate.
\end{proof}

\section{Strong emergent estimates to the reduced Lohe matrix model} \label{sec:6}
\setcounter{equation}{0}
In this section, we present improved emergent dynamics of the reduced Lohe matrix model:
\begin{equation} \label{F-1}
\begin{cases}
\displaystyle {\dot T}_i=A_iT_i+\kappa_1 (T_cT_i^*T_i-T_iT_c^*T_i), \quad t > 0, \\
\displaystyle T_i(0)=U_i^0\Sigma V^*, \quad i = 1, \cdots, N.
\end{cases}
\end{equation}
Recall that
\[  \|U \|^2_F := \sum_{\alpha, \beta} |[U]_{\alpha \beta}|^2 = \mbox{tr}(U^* U), \qquad  {\mathcal D}(U) := \max_{1 \leq i, j \leq N} \|U_i - U_j \|_F. \]
Since this section is rather lengthy compared to other sections, we first briefly summarize our action strategy plan in four steps: \newline
\begin{itemize}
\item
Step A:~First, we provide a reformulation of \eqref{F-1} as the Lohe matrix model with a diagonal frustration (Section \ref{sec:6.1}):
\begin{equation} \label{F-1-1}
 \dot{U}_i =B_iU_i+\kappa_1 (U_cD-U_iDU_c^*U_i),
 \end{equation}
and discus the relations between \eqref{F-1} and \eqref{F-1-1} in terms of the complete aggregation.
\item

\vspace{0.2cm}

Step B:~We derive a differential inequality for the ensemble diameter ${\mathcal D}(U) := \max_{i,j} \|U_i - U_j \|_F$ (Section \ref{sec:6.2}):
\[ -2\kappa_1 \mathcal{A} {\mathcal D}(U)+\kappa_1 \mathcal{A} {\mathcal D}(U)^3- {\mathcal D}(B)  \leq\frac{d}{dt} {\mathcal D}(U) \leq-2\kappa_1 \mathcal{B} {\mathcal D}(U)+\kappa_1 \mathcal{A} {\mathcal D}(U)^3+ {\mathcal D}(B).
\] 
\item
\vspace{0.2cm}

Step C:~For a homogeneous ensemble with ${\mathcal D}(B) = 0$, one has an exponential aggregation (Section \ref{sec:6.3}): there exists a positive constant $\Lambda$ independent of the initial data such that 
\begin{equation*} \label{F-1-2}
 {\mathcal D}(U(t))\leq \mathcal{O}(1) e^{-2\kappa_1 \Lambda t}, \quad \mbox{as $t \to \infty$}. 
 \end{equation*}
Note that in Theorem \ref{T5.4}, we have provided a weak emergent estimate without any decay rate. 
\item
\vspace{0.2cm}
Step D:~For a heterogeneous ensemble with ${\mathcal D}(B) > 0$, one has a practical aggregation (Section \ref{sec:6.4}): if the coupling strength and the initial data satisfy
\[ \kappa_1 \gg {\mathcal D}(B),\quad {\mathcal D}(U^0) \ll 1, \]
then one has practical aggregation:
\[
\lim_{\kappa\rightarrow\infty}\limsup_{t\rightarrow\infty} {\mathcal D}(U)=0.
\]
\end{itemize}
In the following four subsections, we will perform the above steps one by one.
\subsection{A reformulation} \label{sec:6.1}
In this subsection, we reformulate system \eqref{F-1} in terms of unitary matrices $U_i$ as in previous section. First, we substitute the ansatz 
\[
T_i(t)=U_i(t)\Sigma V^*, \quad i = 1, \cdots, N.
\]
into \eqref{F-1} to see 
\begin{equation} \label{F-2}
\dot{U}_i \Sigma V^*=A_i(U_i\Sigma V^*)+\kappa_1 (U_cD-U_iD^*U_c^*U_i)\Sigma V^*
\end{equation}
where $A_i(U_i\Sigma V^*)$ is a tensor contraction and $\Sigma\Sigma^*=D$. To derive a simple form of \eqref{F-2}, we consider the following natural frequency $A_i$:
\begin{equation} \label{F-3}
[A_i]_{\alpha\beta\gamma\delta} :=[B_i]_{\alpha\gamma}\delta_{\beta\delta}.
\end{equation}
Then, it follows from \eqref{F-2} and \eqref{F-3} that 
\[
\dot{U}_i \Sigma V^*=(B_iU_i+\kappa_1 (U_cD-U_iDU_c^*U_i))\Sigma V^*,
\]
This yields the reduced Lohe matrix model of non-identical generalized Lohe matrix model with a diagonal frustration:
\begin{align}\label{F-4}
\begin{cases}
\dot{U}_i =B_iU_i+\kappa_1 (U_cD-U_iDU_c^*U_i), \quad t > 0, \\
U_i(0)=U_i^0, \quad i = 1, \cdots, N.
\end{cases}
\end{align}
\subsection{Evolution of state diameter} \label{sec:6.2}
In this subsection, we study time-evolution of $ {\mathcal D}(U)$.  For this, we assume that $D$ is a diagonal matrix with positive real diagonal entries. For $i, j$, one has
\begin{equation} \label{F-4-1}
\frac{d}{dt}\| U_i - U_j \|_F^2 = \frac{d}{dt}\mathrm{tr}[(U_i-U_j)^*(U_i-U_j)].
\end{equation}
Now, we further simply the R.H.S. of \eqref{F-4-1} as follows.
\begin{align*}
\begin{aligned}
&\frac{d}{dt}\mathrm{tr}[(U_i-U_j)^*(U_i-U_j)] =\frac{d}{dt}\mathrm{tr}[2I-U_i^*U_j-U_j^*U_i]\\
&  \hspace{1cm}  =-\kappa_1 \mathrm{tr}[U_i^*(U_cD-U_jDU_c^*U_j)+(DU_c^*-U_i^*U_cDU_i^*)U_j]+(i\leftrightarrow j)\\
& \hspace{1.4cm}  -\kappa_1 \mathrm{tr}[U_i^*B_jU_j-U_i^*B_iU_j-U_j^*B_jU_i+U_j^*B_iU_i]\\
& \hspace{1cm}  =-\kappa_1 \mathrm{tr}[U_i^*U_cD-U_i^*U_jDU_c^*U_j+DU_c^*U_j-U_i^*U_cDU_i^*U_j]  \\
& \hspace{1.4cm}  -\kappa_1 \mathrm{tr}[U_j^*U_cD-U_j^*U_iDU_c^*U_i+DU_c^*U_i-U_j^*U_cDU_j^*U_i]\\
& \hspace{1.4cm}  -\kappa_1 \mathrm{tr}[U_i^*B_jU_j-U_i^*B_iU_j-U_j^*B_jU_i+U_j^*B_iU_i]\\
& \hspace{1cm}  =-\frac{\kappa_1}{N}\sum_{k=1}^N\big(\mathrm{tr}[U_i^*U_kD-U_i^*U_jDU_k^*U_j+DU_k^*U_j-U_i^*U_kDU_i^*U_j\\
& \hspace{1.4cm} +U_j^*U_kD-U_j^*U_iDU_k^*U_i+DU_k^*U_i-U_j^*U_kDU_j^*U_i]\big)\\
& \hspace{1.4cm} -\kappa_1 \mathrm{tr}[U_i^*B_jU_j-U_i^*B_iU_j-U_j^*B_jU_i+U_j^*B_iU_i].
\end{aligned}
\end{align*}
Next, we define $\mathcal{N}$ and $\mathcal{M}_k$ for all $k=1, 2, \cdots, N$ as follows:
\begin{align}
\begin{aligned} \label{F-6}
\mathcal{N}&:=\mathrm{tr}[U_i^*B_jU_j-U_i^*B_iU_j-U_j^*B_jU_i+U_j^*B_iU_i], \\
\mathcal{M}_k&:=\mathrm{tr}[U_i^*U_kD-U_i^*U_jDU_k^*U_j+DU_k^*U_j-U_i^*U_kDU_i^*U_j \\
&\hspace{1.2cm} U_j^*U_kD-U_j^*U_iDU_k^*U_i+DU_k^*U_i-U_j^*U_kDU_j^*U_i].
\end{aligned}
\end{align}
We set
\[
D :=\mathrm{diag}(\lambda_1^2, \cdots, \lambda_{d_1}^2), \quad 
\langle \lambda^2 \rangle :=\frac{1}{d_1}(\lambda_1^2+\lambda_2^2+\cdots+\lambda_{d_1}^2),\quad \Delta(\lambda^2):=\max_{1\leq k\leq d_1}|\lambda_k^2- \langle \lambda^2 \rangle |.
\]
Then we can decompose the diagonal matrix $D$ as follows:
\[
D=\langle \lambda^2 \rangle I_{d_1} +E,
\]
where $I$ is an identity matrix and $E$ is a diagonal matrix which absolute value of each component is not bigger than $\Delta(\lambda^2)$. In the following lemma, we present some estimates for $\mathcal{N}$ and $\mathcal{M}_k$. 
\begin{lemma} \label{L6.1}
The quantities ${\mathcal N}$ and ${\mathcal M}_k$ in \eqref{F-6} satisfy
\begin{eqnarray*}
&& (i)~\mathcal{N} =\mathrm{tr}[(B_j-B_i)(U_jU_i^*-U_iU_j^*)], \\
&& (ii)~\mathcal{M}_k = 4\langle \lambda^2 \rangle \mathrm{tr}[(U_i-U_j)^*(U_i-U_j)] -\langle \lambda^2 \rangle \mathrm{tr}[(U_i-U_k)^*(U_i-U_k)(U_i-U_j)^*(U_i-U_j) \\
&& +(U_j-U_k)^*(U_j-U_k)(U_i-U_j)^*(U_i-U_j)] +4\mathrm{tr}[E(U_i-U_j)^*(U_i-U_j)]\\
&& + \mathrm{tr}[((I-U_i^*U_k)E+E(I-U_k^*U_j)+(I-U_j^*U_k)E+E(I-U_k^*U_i))(U_i-U_j)^*(U_i-U_j)].
\end{eqnarray*}
\end{lemma}
\begin{proof}
(i) We rearrange the terms inside the bracket in $\mathcal{N}$ and use the properties $\mathrm{tr}(AB)=\mathrm{tr}(BA)$ and $U_i, U_j\in \mathbb{U}(d)$ to yield
\[
\mathcal{N} =\mathrm{tr}[(B_j-B_i)(U_jU_i^*-U_iU_j^*)].
\]
(ii)~To simplify the term ${\mathcal M}_k$, we take the following two steps. \newline

\noindent $\bullet$~Step A: We claim
\begin{align}
\begin{aligned}  \label{F-6-1}
\mathcal{M}_k &= 4\mathrm{tr} \Big [D(U_i-U_j)^*(U_i-U_j)]-\mathrm{tr}[((I-U_i^*U_k)D+D(I-U_k^*U_j) \\
& \hspace{1.5cm} +(I-U_j^*U_k)D +D(I-U_k^*U_i))(U_i-U_j)^*(U_i-U_j) \Big ].
\end{aligned}
\end{align}
For the derivation of \eqref{F-6-1}, one has
\begin{align*}
\mathcal{M}_k&=\mathrm{tr}[U_i^*U_kD(I-U_i^*U_j)+(I-U_i^*U_j)DU_k^*U_j+U_j^*U_kD(I-U_j^*U_i)+(I-U_j^*U_i)DU_k^*U_i]\\
&=\mathrm{tr}[(U_i^*U_kD+DU_k^*U_j)(I-U_i^*U_j)+(U_j^*U_kD+DU_k^*U_i)(I-U_j^*U_i)]\\
&=\mathrm{tr}[(U_i^*U_kD+DU_k^*U_j)U_i^*(U_i-U_j)+(U_j^*U_kD+DU_k^*U_i)U_j^*(U_j-U_i)]\\
&=\mathrm{tr}[\underbrace{(U_i^*U_kDU_i^*+DU_k^*U_jU_i^*-U_j^*U_kDU_j^*-DU_k^*U_iU_j^*)(U_i-U_j)}_{:=\mathcal{J}_1}].
\end{align*}
We further estimate the terms $\mathcal{J}_1$ as follows:
\begin{align}
\begin{aligned} \label{F-7}
\mathcal{J}_1&=U_i^*U_kD(U_i-U_j)^*(U_i-U_j)+DU_k^*U_j(U_i-U_j)^*(U_i-U_j) -U_j^*U_kD(U_j-U_i)^*(U_i-U_j)\\
&-DU_k^*U_i(U_j-U_i)^*(U_i-U_j)+U_i^*U_kDU_j^*(U_i-U_j)-U_j^*U_kDU_i^*(U_i-U_j)\\
&=\underbrace{(U_i^*U_kD+DU_k^*U_j+U_j^*U_kD+DU_k^*U_i)(U_i-U_j)^*(U_i-U_j)}_{=: {\mathcal J}_{11}} + \underbrace{(U_i^*U_kDU_j^*-U_j^*U_kDU_i^*)(U_i-U_j)}_{=: {\mathcal J}_{12}}.
\end{aligned}
\end{align}

\noindent $\diamond$~(Estimate of $\mathcal{J}_{11}$): By direct estimate, one has
\begin{align}
\begin{aligned} \label{F-8}
\mathcal{J}_{11}&=(U_i^*U_kD+DU_k^*U_j+U_j^*U_kD+DU_k^*U_i)(U_i-U_j)^*(U_i-U_j)\\
&=4D(U_i-U_j)^*(U_i-U_j) -((I-U_i^*U_k)D+D(I-U_k^*U_j) \\
& \hspace{0.2cm} +(I-U_j^*U_k)D+D(I-U_k^*U_i))(U_i-U_j)^*(U_i-U_j).
\end{aligned}
\end{align}

\noindent $\diamond$~(Estimate of $\mathcal{J}_{12}$): Similarly, one has
\begin{align}
\begin{aligned} \label{F-9}
\mathrm{tr}[\mathcal{J}_{12}]&=\mathrm{tr}[(U_i^*U_kDU_j^*-U_j^*U_kDU_i^*)(U_i-U_j)]\\
&=\mathrm{tr}[U_kDU_j^*-U_i^*U_kD-U_j^*U_kD+U_kDU_i^*]=0.
\end{aligned}
\end{align}
Finally in \eqref{F-7}, we combine \eqref{F-8} and \eqref{F-9} to get \eqref{F-6-1}. \newline

\noindent $\bullet$~Step B: We can rewrite $\mathcal{M}_k$ as follows
\begin{align*}
\mathcal{M}_k&=4\mathrm{tr}[(\langle \lambda^2 \rangle I +E)(U_i-U_j)^*(U_i-U_j)]\\
&-\mathrm{tr}[((I-U_i^*U_k)(\overline{\lambda^2}I+E)+(\langle \lambda^2 \rangle I +E)(I-U_k^*U_j)\\
&+(I-U_j^*U_k)(\langle \lambda^2 \rangle I+E)+(\langle \lambda^2 \rangle I+E)(I-U_k^*U_i))(U_i-U_j)^*(U_i-U_j)]\\
&=4\langle \lambda^2 \rangle \mathrm{tr}[(U_i-U_j)^*(U_i-U_j)]\\
&-\langle \lambda^2 \rangle \mathrm{tr}[((I-U_i^*U_k)+(I-U_k^*U_j)+(I-U_j^*U_k)+(I-U_k^*U_i))(U_i-U_j)^*(U_i-U_j)]\\
&+4\mathrm{tr}[E(U_i-U_j)^*(U_i-U_j)]\\
&-\mathrm{tr}[((I-U_i^*U_k)E+E(I-U_k^*U_j)+(I-U_j^*U_k)E+E(I-U_k^*U_i))(U_i-U_j)^*(U_i-U_j)]\\
&=4\langle \lambda^2 \rangle \mathrm{tr}[(U_i-U_j)^*(U_i-U_j)] -\langle \lambda^2 \rangle \mathrm{tr}[(U_i-U_k)^*(U_i-U_k)(U_i-U_j)^*(U_i-U_j) \\
&+(U_j-U_k)^*(U_j-U_k)(U_i-U_j)^*(U_i-U_j)] +4\mathrm{tr}[E(U_i-U_j)^*(U_i-U_j)]\\
&-\mathrm{tr}[((I-U_i^*U_k)E+E(I-U_k^*U_j)+(I-U_j^*U_k)E+E(I-U_k^*U_i))(U_i-U_j)^*(U_i-U_j)]
\end{align*}
to get the desired estimate.
\end{proof}

\begin{lemma} \label{L6.2}
The following assertions hold.
\begin{enumerate}
\item
Let $E$ be a $d_1 \times d_1$ diagonal matrix with $\max_{i}|[E]_{ii}|=\varepsilon$. Then one has
\[
|\mathrm{tr}[AE]| \leq \varepsilon\cdot|\mathrm{tr}[A]|.
\]
\item
The following estimates hold.
\[ |\mathrm{tr}[ABCD]|\leq||A||_F\cdot||B||_F\cdot||C||_F\cdot||D||_F, \quad |\mathrm{tr}[ABC]|\leq||A||_F\cdot||B||_F\cdot||C||_F. \]
\end{enumerate}
\end{lemma}
\begin{proof} 
(i)~By definition of trace, one has
\[
\mathrm{tr}[AE]=\sum_{\alpha\beta}[A]_{\alpha\beta}[E]_{\alpha\beta}=\sum_{\alpha}[A]_{\alpha\alpha}[E]_{\alpha\alpha}.
\]
This yields
\[
\Big |\mathrm{tr}[AE] \Big|=\Big|\sum_{\alpha}[A]_{\alpha\alpha}[E]_{\alpha\alpha} \Big|\leq \varepsilon\Big|\sum_{\alpha}[A]_{\alpha\alpha} \Big|=\varepsilon\cdot|\mathrm{tr}[A]|.
\]
(ii)~By direct estimate, one has
\[
\mathrm{tr}[ABCD]=\sum_{\alpha,\beta,\gamma,\delta}[A]_{\alpha\beta}[B]_{\beta\gamma}[C]_{\gamma\delta}[D]_{\delta\alpha}=\sum_{\alpha,\beta,\gamma,\delta}([A]_{\alpha\beta}[C]_{\gamma\delta})([B]_{\beta\gamma}[D]_{\delta\alpha}).
\]
We can apply the Cauchy-Schwarz inequality to above equality:
\begin{align*}
\begin{aligned}
\Big|\mathrm{tr}[ABCD] \Big|^2&=\left | \sum_{\alpha,\beta,\gamma,\delta}([A]_{\alpha\beta}[C]_{\gamma\delta})([B]_{\beta\gamma}[D]_{\delta\alpha})\right|^2  \\&\leq\left(\sum_{\alpha,\beta,\gamma,\delta} \Big| [A]_{\alpha\beta}[C]_{\gamma\delta} \Big|^2\right)\left(\sum_{\alpha,\beta,\gamma,\delta}\Big| [B]_{\beta\gamma}[D]_{\delta\alpha} \Big|^2\right)\\
&=\Big (||A||_F\cdot||B||_F\cdot||C||_F\cdot||D||_F \Big)^2.
\end{aligned}
\end{align*}

Similarly, one has the second estimate.
\end{proof}

\begin{lemma} \label{L6.3}
The term ${\mathcal M}_k$ satisfies 
\begin{align*}
\begin{aligned}
&4(\langle \lambda^2 \rangle -\Delta(\lambda^2))\mathrm{tr} \Big [(U_i-U_j)^*(U_i-U_j) \Big] -(\langle \lambda^2 \rangle +\Delta(\lambda^2)) \Big| \mathrm{tr} \Big[(U_i-U_k)^*(U_i-U_k)(U_i-U_j)^*(U_i-U_j) \\
& +(U_j-U_k)^*(U_j-U_k)(U_i-U_j)^*(U_i-U_j) \Big]  \Big| \leq\mathcal{M}_k \leq4(\langle \lambda^2 \rangle +\Delta(\lambda^2))\mathrm{tr}[(U_i-U_j)^*(U_i-U_j)]\\
&+(\langle \lambda^2 \rangle +\Delta(\lambda^2))|\mathrm{tr}[(U_i-U_k)^*(U_i-U_k)(U_i-U_j)^*(U_i-U_j) +(U_j-U_k)^*(U_j-U_k)(U_i-U_j)^*(U_i-U_j)]|.
\end{aligned}
\end{align*}
\end{lemma}
\begin{proof} We use Lemma \ref{L6.2} to simplify $\mathcal{M}_k$. For this, we set $\varepsilon=\Delta(\lambda^2)$ to get 
\begin{align*}
|\mathrm{tr}[E(U_i-U_j)^*(U_i-U_j)]|\leq \Delta(\lambda^2)|\mathrm{tr}[(U_i-U_j)^*(U_i-U_j)]|.
\end{align*}
On the other hand, we also use $\mbox{tr}(AB) = \mbox{tr}(BA)$ to obtain
\begin{align*}
\begin{aligned}
&|\mathrm{tr}[((I-U_i^*U_k)E+E(I-U_k^*U_j)+(I-U_j^*U_k)E+E(I-U_k^*U_i))(U_i-U_j)^*(U_i-U_j)]|\\
& \hspace{0.2cm} \leq \Delta(\lambda^2)|\mathrm{tr}[((I-U_i^*U_k)+(I-U_k^*U_j)+(I-U_j^*U_k)+(I-U_k^*U_i))(U_i-U_j)^*(U_i-U_j)]|\\
& \hspace{0.2cm} =\Delta(\lambda^2) \Big|\mathrm{tr}[(U_i-U_k)^*(U_i-U_k)(U_i-U_j)^*(U_i-U_j) +(U_j-U_k)^*(U_j-U_k)(U_i-U_j)^*(U_i-U_j)] \Big |.
\end{aligned}
\end{align*}
Then, we use the property:
\[
0\leq\mathrm{tr}[(U_i-U_j)^*(U_i-U_j)]=||U_i-U_j||_F^2
\]
to yield the desired estimates.
\end{proof}

\begin{proposition} \label{P6.1}
Let $\{U_i \} \subset \bbu(d_1)$ be a solution to \eqref{F-4}. Then, one has 
\[ -2\kappa_1 \mathcal{A} {\mathcal D}(U)+\kappa_1 \mathcal{A} {\mathcal D}(U)^3- {\mathcal D}(B)  \leq\frac{d}{dt} {\mathcal D}(U) \leq-2\kappa_1 \mathcal{B} {\mathcal D}(U)+\kappa_1 \mathcal{A} {\mathcal D}(U)^3+ {\mathcal D}(B)
\]
where
\[
\mathcal{A} :=\langle \lambda^2 \rangle +\Delta(\lambda^2),\quad\mathcal{B} :=\langle \lambda^2 \rangle -\Delta(\lambda^2).
\]
\end{proposition}
\begin{proof}
It follows from the estimate (ii) in Lemma \ref{L6.2} that 
\begin{align*}
\begin{aligned}
&|\mathrm{tr}[(U_i-U_k)^*(U_i-U_k)(U_i-U_j)^*(U_i-U_j)+(U_j-U_k)^*(U_j-U_k)(U_i-U_j)^*(U_i-U_j)]|\\
& \hspace{1cm} \leq(||U_i-U_k||_F^2+||U_j-U_k||^2)||U_i-U_j||_F^2,
\end{aligned}
\end{align*}
and
\begin{align*}
|\mathcal{N}|&=|\mathrm{tr}[(B_j-B_i)(U_j-U_i)(U_j^*+U_i^*)]|\leq||B_i-B_j||_F\cdot||U_i-U_j||_F\cdot||U_j+U_i||_F\\
&\leq2 \sqrt{d_1} ||B_i-B_j||_F\cdot||U_i-U_j||_F,
\end{align*}
where we used 
\[ \|U_i \|_F = \sqrt{d_1} \quad \mbox{and} \quad  \|U_j \|_F = \sqrt{d_1}. \]

Recall that 
\[
{\mathcal D}(U)=\max_{i, j}||U_i-U_j||_F, \quad  {\mathcal D}(B)=\max_{i, j}||B_i-B_j||_F
\]
Then, one has
\begin{align}\label{F-10}
\begin{aligned}
&\Big |\mathrm{tr}[(U_i-U_k)^*(U_i-U_k)(U_i-U_j)^*(U_i-U_j) \\
& \hspace{1cm} +(U_j-U_k)^*(U_j-U_k)(U_i-U_j)^*(U_i-U_j)] \Big| \leq2 {\mathcal D}(U)^4.
\end{aligned}
\end{align}
On the other hand, it follows from \eqref{F-10} that 
\begin{align*}
\frac{d}{dt}||U_i-U_j||_F^2=-\frac{\kappa}{N}\sum_{k=1}^N\mathcal{M}_k-\kappa \mathcal{N}.
\end{align*}
Now, we use Lemma \ref{L6.3} and  \eqref{F-10} to obtain
\begin{align*}
\begin{aligned}
&4(\langle \lambda^2 \rangle -\Delta(\lambda^2))||U_i-U_j||_F^2-2(\langle \lambda^2 \rangle +\Delta(\lambda^2)) {\mathcal D}(U)^4 \\
& \hspace{2cm} \leq\mathcal{M}_k \leq (\langle \lambda^2 \rangle +\Delta(\lambda^2))(4||U_i-U_j||_F^2+2 {\mathcal D}(U)^4).
\end{aligned}
\end{align*}
Thus, one has
\begin{align*}
&-\kappa_1 (\langle \lambda^2 \rangle +\Delta(\lambda^2))(4||U_i-U_j||_F^2+2 {\mathcal D}(U)^4)-2||B_i-B_j||_F\cdot||U_i-U_j||_F \\
& \hspace{1cm} \leq \frac{d}{dt}||U_i-U_j||_F^2\leq -4\kappa_1 (\langle \lambda^2 \rangle -\Delta(\lambda^2))||U_i-U_j||_F^2+2\kappa_1 (\langle \lambda^2 \rangle +\Delta(\lambda^2)) {\mathcal D}(U)^4 \\
& \hspace{1.5cm} +2||B_i-B_j||_F\cdot||U_i-U_j||_F.
\end{align*}
Since above inequality holds for all $i, j$, we obtain
\begin{align*}
\begin{aligned}
&-\kappa_1 (\langle \lambda^2 \rangle +\Delta(\lambda^2))(4{\mathcal D}(U)^2+2{\mathcal D}(U)^4)-2{\mathcal D}(U) {\mathcal D}(B) \\
& \hspace{1.5cm} \leq\frac{d}{dt} {\mathcal D}(U)^2 \leq-4\kappa_1 (\langle \lambda^2 \rangle -\Delta(\lambda^2)) {\mathcal D}(U)^2+2\kappa_{1}(\langle \lambda^2 \rangle +\Delta(\lambda^2)) {\mathcal D}(U)^4+2 {\mathcal D}(U) {\mathcal D}(B).
\end{aligned}
\end{align*}
This yields
\begin{align}
\begin{aligned} \label{F-10-1}
&-\kappa_1 (\langle \lambda^2 \rangle +\Delta(\lambda^2))(2{\mathcal D}(U)+ {\mathcal D}(U)^3)-{\mathcal D}(B) \\
& \hspace{1cm} \leq \frac{d}{dt} {\mathcal D}(U)\leq-2\kappa_1 (\langle \lambda^2 \rangle -\Delta(\lambda^2)) {\mathcal D}(U)+\kappa_1 (\langle \lambda^2 \rangle +\Delta(\lambda^2)) {\mathcal D}(U)^3+ {\mathcal D}(B).
\end{aligned}
\end{align}
Now we set
\begin{equation} \label{F-10-2}
\mathcal{A} :=\langle \lambda^2 \rangle +\Delta(\lambda^2),\quad\mathcal{B} :=\langle \lambda^2 \rangle -\Delta(\lambda^2).
\end{equation}
Note that ${\mathcal A}$ and ${\mathcal B}$ are determined by the initial data.  Finally, we combine \eqref{F-10-1} and \eqref{F-10-2} to get the desired result.
\end{proof}
In the following two subsections, we consider the cases:
\[ \mbox{Either} \quad {\mathcal D}(B) = 0 \quad \mbox{or} \quad {\mathcal D}(B)  > 0. \]
\subsection{Exponential aggregation} \label{sec:6.3}
Consider the case ${\mathcal D}(B)=0$. In this case, the differential inequality in Proposition \ref{P6.1} implies  
\begin{equation} \label{F-11}
-2\kappa_1 \mathcal{A} {\mathcal D}(U)+\kappa_1 \mathcal{A} {\mathcal D}(U)^3\leq\frac{d}{dt} {\mathcal D}(U)\leq-2\kappa_1 \mathcal{B}{\mathcal D}(U)+\kappa_1 \mathcal{A} {\mathcal D}(U)^3, \quad \mbox{a.e.}~~t > 0.
\end{equation}
We set 
\[  X={\mathcal D}(U)^2. \]
 Then, it follows from \eqref{F-11} that
\begin{equation} \label{F-12}
-4\kappa_1 \mathcal{A}X+2\kappa_1 \mathcal{A}X^2\leq \frac{dX}{dt} \leq-4\kappa_1 \mathcal{B}X+2\kappa_1 \mathcal{A}X^2.
\end{equation}
Suppose that the initial data $X(0)=X_0$ satisfy
\[
0\leq X_0< \frac{2\mathcal{B}}{\mathcal{A}}\leq2.
\]
By integrating  \eqref{F-12}, one has 
\[
\frac{2X_0}{X_0+(2-X_0)e^{4\kappa_1 \mathcal{A}t}}\leq X \leq\frac{2\mathcal{B}}{\mathcal{A}}\cdot\frac{X_0}{X_0+\left(\frac{2\mathcal{B}}{\mathcal{A}}-X_0\right)e^{4\kappa_1 \mathcal{B}t}}.
\]
From above inequality we can obtain following theorem.
\begin{theorem}\label{T6.1}
Suppose that $B_i,~\{U_i^0 \}$ and diagonal frustration matrix $D$ satisfy
\[ B_i=0, \quad i = 1, \cdots, N, \qquad {\mathcal D}(U^0)\leq \sqrt{\frac{2\mathcal{B}}{\mathcal{A}}}, \quad  D=\mathrm{diag}(\lambda_1^2, \cdots, \lambda_{d_1}^2), \quad \mathcal{B}=\langle{\lambda^2}\rangle-\Delta(\lambda^2)>0.
\]
Then for any solution $\{U_i\}_{i=1}^N$ to system \eqref{F-4}, we have an exponential aggregation:
\[
\sqrt{\frac{2{\mathcal D}^2(U^0)}{{\mathcal D}^2(U^0)+(2-{\mathcal D}^2(U^0))e^{4\kappa_1 \mathcal{A}t}}} \leq {\mathcal D}(U(t))\leq \sqrt{\frac{2\mathcal{B}}{\mathcal{A}}\cdot\frac{{\mathcal D}^2(U^0)}{{\mathcal D}^2(U^0)+\left(\frac{2\mathcal{B}}{\mathcal{A}}-{\mathcal D}^2(U^0)\right)e^{4\kappa_1 \mathcal{B}t}}}.
\]
\end{theorem}
\begin{remark}
Theorem \ref{T6.1} implies 
\[
\mathcal{O}(1)e^{-2\kappa_1 \mathcal{A}t}\leq {\mathcal D}(U(t))\leq \mathcal{O}(1) e^{-2\kappa_1 \mathcal{B}t}, \quad \mbox{as $t \to \infty$}.
\]
\end{remark}
As a corollary of Theorem \ref{T6.1}, we have the complete aggregation of the generalized Lohe matrix model:
\begin{equation*} \label{F-13}
\begin{cases}
{\dot T}_i=\kappa_{1}(T_cT_i^* T_i-T_iT_c^* T_i), \quad t > 0,\\
T_i(0)=T_i^0=U_i^0\Sigma V^*, \quad i = 1, \cdots, N.
\end{cases}
\end{equation*}
Recall that 
\[ D = \Sigma\Sigma^*=\mathrm{diag}(\lambda_1^2, \cdots,\lambda_{d_1}^2). \]
\begin{corollary} \label{C6.1}
Suppose that the initial data and diagonal frustration satisfy
\[
{\mathcal D}(T^0)<\min_k\lambda_k\cdot\sqrt{\frac{2\mathcal{B}}{\mathcal{A}}}, \quad \lambda_i^2 > 0, \quad i = 1, \cdots, d_1, \quad \mathcal{B}=\langle \lambda^2 \rangle -\Delta(\lambda^2)>0,
\]
and let $\{T_i \}$ be a solution to system \eqref{F-1}. Then, we have
\[
\mathcal{O}(1)e^{-2\kappa_{1}\mathcal{A}t}\leq {\mathcal D}(T(t))\leq \mathcal{O}(1) e^{-2\kappa_{1}\mathcal{B}t}.
\]
\end{corollary}
\begin{proof}
Since $D=\Sigma\Sigma^*$ has no zero diagonal entry, we can directly apply Theorem \ref{T6.1}.  Next, we need to find the condition of ${\mathcal D}(T^0)$. Now, we use $T_i^0=U_i^0\Sigma V^*$ to see that 
\begin{align*}
||T_i^0-T_j^0||_F^2&=\mathrm{tr}[(T_i^0-T_j^0)^*(T_i^0-T_j^0)]=\mathrm{tr}[V\Sigma^*(U_i^0-U_j^0)^*(U_i^0-U_j^0)\Sigma V^*]\\
&=\mathrm{tr}[D(U_i^0-U_j^0)^*(U_i^0-U_j^0)].
\end{align*}
This yields
\[
\Big( \min_k\lambda_k \Big) ||U_i^0-U_j^0||_F\leq ||T_i^0-T_j^0||_F\leq \Big( \max_k\lambda_k \Big) ||U_i^0-U_j^0||_F.
\]
If the initial data satisfy
\[
{\mathcal D}(T^0)< \Big( \min_k\lambda_k \Big) \cdot\sqrt{\frac{2\mathcal{B}}{\mathcal{A}}},
\]
then we have
\[
||U_i^0-U_j^0||^2_F\leq\frac{1}{\Big( \min_k\lambda_k\Big)^2}||T_i^0-T_j^0||_F^2<\frac{2\mathcal{B}}{\mathcal{A}}.
\]
\end{proof}

\subsection{Practical aggregation} \label{sec:6.4}
In this subsection, we consider  the case: 
\[ {\mathcal D}(B)>0, \quad \mathcal{A} > 0 \quad \mbox{and} \quad \mathcal{B} > 0. \]
We begin with following inequality:
\[
\frac{d}{dt} {\mathcal D}(U)\leq {\mathcal D}(B)-2\kappa_1 \mathcal{B} {\mathcal D}(U)+\kappa_1 \mathcal{A} {\mathcal D}(U)^3, \quad \mbox{a.e.}~t \in (0, \infty).
\]
Consider the following cubic polynomial:
\[
f(x)=2\mathcal{B}x-\mathcal{A}x^3,\quad x\geq0.
\]
Then, the upper bound of $\frac{d}{dt} {\mathcal D}(U)$ can be expressed as 
\[
\frac{d}{dt} {\mathcal D}(U)\leq {\mathcal D}(B)-\kappa_1 f( {\mathcal D}(U)),\quad \mbox{a.e.}~~t>0.
\]
In the sequel, we study some properties of some polynomial $f$. \newline

\noindent We set 
\[
g(x) :=\frac{{\mathcal D}(B)}{\kappa_1}-f(x)=\mathcal{A}x^3-2\mathcal{B}x+\frac{{\mathcal D}(B)}{\kappa_1}, \quad  x\geq0.
\]
Then one has 
\[
x_m=\sqrt{\frac{2\mathcal{B}}{3\mathcal{A}}} = \mbox{argmin}_{x \geq 0} g(x), \qquad g(x_m)=-\sqrt{\frac{32\mathcal{B}^3}{27\mathcal{A}}}+\frac{{\mathcal D}(B)}{\kappa_1}.
\]
\begin{lemma}
Suppose that the coupling strength satisfies
\begin{equation} \label{F-13-1}
\kappa_1 > {\mathcal D}(B)\cdot\sqrt{\frac{27\mathcal{A}}{32\mathcal{B}^3}}.
\end{equation}
Then, there exist two distinct positive roots $0 < \alpha_1 < \alpha_2$ of $g$ such that 
\begin{eqnarray*}
&& (i)~ g(x)>0 \quad \mbox{for}~~x\in[0, \alpha_1)\cup(\alpha_2, \infty]; \qquad g(x)<0 \quad \mbox{for}~~x\in(\alpha_1, \alpha_2).\\
&& (ii)~0<\alpha_1<\beta :=\frac{3D(B)}{4\mathcal{B}\kappa_1}<x_m=\sqrt{\frac{2\mathcal{B}}{3\mathcal{A}}}<\alpha_2.
\end{eqnarray*}
\end{lemma}
\begin{proof}
(i)~Since 
\[  g(0) =  \frac{{\mathcal D}(B)}{\kappa_1}, \quad g^{\prime} = 3\mathcal{A}x^2-2\mathcal{B},    \]
it is easy to see that 
\[ g^{\prime}(x) = 0 \quad \Longleftrightarrow \quad x = \pm \sqrt{\frac{2\mathcal{B} }{ 3\mathcal{A}}}, \]
and 
\[ g\Big(   \sqrt{\frac{2\mathcal{B} }{ 3\mathcal{A}}} \Big) = -\frac{4\mathcal{B}}{3} \sqrt{\frac{2\mathcal{B} }{ 3\mathcal{A}}}  + \frac{{\mathcal D}(B)}{\kappa_1} < 0 \quad \Longleftrightarrow \quad  \kappa_1 > {\mathcal D}(B)\cdot\sqrt{\frac{27\mathcal{A}}{32\mathcal{B}^3}}.  \]
Using a graphical method, it is easy to see that there exist two positive roots $\alpha_1 < \alpha_2$ of $g(x) = 0$ such that
\[ g(x)>0 \quad \mbox{for}~~x\in[0, \alpha_1)\cup(\alpha_2, \infty]; \qquad g(x)<0 \quad \mbox{for}~~x\in(\alpha_1, \alpha_2). \]
\noindent (ii)~First note that $g$ is convex for $x>0$:
\[
g''(x)=6\mathcal{A}x>0.
\]
We set
\[ P_1=(0, g(0)), \quad P_2=(x_m, g(x_m)), \]
and let $l$ be a straight line passing through two points $P_1$ and $P_2$:
\[
l: y=-\frac{4\mathcal{B}}{3}x+\frac{{\mathcal D}(B)}{\kappa_1}.
\]
Let $Q(\beta, 0)$ be the intersection point of two lines $l$ and $x$-axis:
\[ \beta=\frac{3{\mathcal D}(B)}{4\mathcal{B}\kappa_{1}}. \]
Then it is also easy to see that 
\begin{equation} \label{F-14}
 \alpha_1<x_m<\alpha_2.
\end{equation} 
Since the graph of $y=g(x)$ is convex, the point $P_3(\alpha_1, 0) $ lies below the line $l$. So we can conclude that
\begin{equation} \label{F-15}
 \alpha_1<\beta. 
\end{equation} 
On the other hand, by the assumption \eqref{F-13-1}, one has 
\begin{equation} \label{F-16}
\beta<x_m.
\end{equation}
Finally, we combine all the estmates \eqref{F-14}, \eqref{F-15} and \eqref{F-16} to get
\[
0<\alpha_1<\beta=\frac{3{\mathcal D}(B)}{4\mathcal{B}\kappa_1}<x_m=\sqrt{\frac{2\mathcal{B}}{3\mathcal{A}}}<\alpha_2.
\]
\end{proof}
\begin{theorem}\label{T6.2}
Suppose that the coupling strength and initial data satisfy
\[ \kappa_1 > {\mathcal D}(B)\cdot\sqrt{\frac{27\mathcal{A}}{32\mathcal{B}^3}},\qquad {\mathcal D}(U^0)<\alpha_2, \]
where $\alpha_2$, $\mathcal{A}$ and $\mathcal{B}$ were defined previous theorem and satisfy $\mathcal{B}>0,$
and let $\{U_i\}_{i=1}^N$ be a solution to \eqref{F-4}. Then, one has practical aggregation:
\[
\lim_{\kappa_1\rightarrow\infty}\limsup_{t\rightarrow\infty} {\mathcal D}(U)=0.
\]
\end{theorem}
\begin{proof} By comparison theorem for the differential inequality
\[
\frac{d}{dt} {\mathcal D}(U)\leq\kappa_1g({\mathcal D}(U)),
\]
since ${\mathcal D}(U^0)<\alpha_2$, there exist $t_*>0$ such that
\[
{\mathcal D}(U(t))\leq\alpha_1,\quad t> t_*.
\]
We also have the inequality on $\alpha_1$:
\[
{\mathcal D}(U(t))\leq\alpha_1<\frac{3{\mathcal D}(B)}{4\mathcal{B}\kappa_1}, \quad \forall~t > t_*.
\]
This yields
\[
\limsup_{t\rightarrow\infty} {\mathcal D}(U)\leq\frac{3{\mathcal D}(B)}{4\mathcal{B}\kappa_1}.
\]
Letting $\kappa \to \infty$, one has 
\[
\lim_{\kappa_1 \rightarrow\infty}\limsup_{t\rightarrow\infty} {\mathcal D}(U)=0.
\]
\end{proof}

As a corollary of Theorem \ref{T6.2}, one has practical synchronization for system:
\begin{equation} \label{F-17}
\begin{cases}
{\dot T}_i=B_iT_i+\kappa_1 (T_cT_i^* T_i-T_iT_c^* T_i), \quad t > 0, \\
T_i(0)=T_i^0=U_i^0\Sigma V^*, \quad i = 1, \cdots, N.
\end{cases}
\end{equation}
\begin{corollary}
Suppose that the coupling strength and initial data satisfy
\[
\kappa_1 > {\mathcal D}(B)\cdot\sqrt{\frac{27\mathcal{A}}{32\mathcal{B}^3}}, \qquad  {\mathcal D}(T^0)< \Big( \min_{1 \leq k \leq N} \lambda_k \Big) \cdot\alpha_2, \qquad \mathcal{B}>0,
\]
where $\alpha_2$, $\mathcal{A}$ and $\mathcal{B}$ were defined in Theorem \ref{T6.2}. Then, we have a practical synchronization:
\[
\lim_{\kappa_1\rightarrow\infty}\limsup_{t\rightarrow\infty} {\mathcal D}(T)=0.
\]
\end{corollary}
\begin{proof}
We use ${\mathcal D}(T) \leq {\mathcal O}(1){\mathcal D}(U)$ and Theorem \ref{T6.2} to get the desired estimate.
\end{proof}
\section{Conclusion} \label{sec:7}
\setcounter{equation}{0} 
In this paper, we proposed a generalized Lohe matrix model on a space of complex matrices with the same size motivated by the Lohe tensor model. For the same natural frequency tensor, our proposed model admits the solution splitting property which means that the general solution is the composition of free flow and nonlinear flow. In previous literature, all matrix-valued aggregation models are defined on the space of square matrices. Thus, as far as the authors know, there are no aggregation models for the space of non-square matrices. For our proposed model, we provide several sufficient frameworks leading to the complete aggregation for a homogeneous ensemble, whereas for a heterogeneous ensemble with distributed natural frequencies, we show that if the coupling strength is sufficiently large and initial data is sufficiently aggregated, then we will show that the ensemble diameter can be made sufficiently small, as we increase the coupling strength (emergence of practical synchronization). There are several issues which were not discussed in this paper, for example, gradient flow formulation, emergence of aggregation of generic initial data and formation of local aggregation, etc. These interesting issues will be treated in future works.


\end{document}